\documentclass[runningheads]{llncs}

\usepackage{hyperref}
\usepackage{amsmath}
\usepackage{graphicx}
\usepackage{subfigure}

\newcommand{\CCS}{cubie cluster}
\newcommand{\CCSs}{cubie clusters}
\newcommand{\CCSconfig}{cluster configuration}
\newcommand{\CCSconfigs}{cluster configurations}

\newcommand{\floor}[1]{\ensuremath{\lfloor #1 \rfloor}}
\newcommand{\group}[1]{\ensuremath{\{0, 1, \ldots, #1 - 1\}}}
\newcommand{\concept}[1]{\textit{#1}}
\newcommand{\degree}{\ensuremath{^\circ}}

\newcommand{\COMP}{\ensuremath{\circ}}

\newcommand{\SEQ}[1]{\textsc{Seq}(#1)}
\newcommand{\LS}[1]{\textsc{LSeq}(#1)}

\newcommand{\shortmoves}{\ensuremath{2^{2c_1c_2 + 8(c_1 + c_2)} - 1}}
\newcommand{\longmoves}{\ensuremath{(c_1c_2)! \cdot 2^{1 + 3c_1c_2 + 8(c_1 + c_2)}}}
\newcommand{\tourmoves}{\ensuremath{(c_1c_2)! \cdot 2^{1 + c_1c_2}}}

{\makeatletter
 \gdef\xxxmark{%
   \expandafter\ifx\csname @mpargs\endcsname\relax 
     \expandafter\ifx\csname @captype\endcsname\relax 
       \marginpar{xxx}
     \else
       xxx 
     \fi
   \else
     xxx 
   \fi}
 \gdef\xxx{\@ifnextchar[\xxx@lab\xxx@nolab}
 \long\gdef\xxx@lab[#1]#2{\textbf{[\xxxmark #2 ---{\sc #1}]}}
 \long\gdef\xxx@nolab#1{\textbf{[\xxxmark #1]}}
}

\newif\ifabstract
\abstracttrue
\abstractfalse
\newif\iffull
\ifabstract \fullfalse \else \fulltrue \fi

\let\realbfseries=\bfseries
\def\bfseries{\realbfseries\boldmath}

\newcounter{theorem-preserve}
\newcounter{lemma-preserve}
\newcounter{definition-preserve}
\newcounter{corollary-preserve}
\newtoks\magicAppendix
\magicAppendix={}
\newtoks\magictoks
\newif\iflater
\laterfalse
\ifabstract
\long\def\later#1{\magictoks={#1}%
  \edef\magictodo{\noexpand\magicAppendix={\the\magicAppendix \par
    \the\magictoks}}%
  \magictodo}
\long\def\both#1{\magictoks={#1}%
  \edef\magictodo{\noexpand\magicAppendix={\the\magicAppendix \par
    \noexpand\setcounter{theorem-preserve}{\noexpand\arabic{theorem}}%
    \noexpand\setcounter{lemma-preserve}{\noexpand\arabic{lemma}}%
    \noexpand\setcounter{definition-preserve}{\noexpand\arabic{definition}}%
    \noexpand\setcounter{corollary-preserve}{\noexpand\arabic{corollary}}%
    \noexpand\setcounter{theorem}{\arabic{theorem}}%
    \noexpand\setcounter{lemma}{\arabic{lemma}}%
    \noexpand\setcounter{definition}{\arabic{definition}}%
    \noexpand\setcounter{corollary}{\arabic{corollary}}%
    \the\magictoks%
    \noexpand\setcounter{theorem}{\noexpand\arabic{theorem-preserve}}%
    \noexpand\setcounter{lemma}{\noexpand\arabic{lemma-preserve}}%
    \noexpand\setcounter{definition}{\noexpand\arabic{definition-preserve}}%
    \noexpand\setcounter{corollary}{\noexpand\arabic{corollary-preserve}}}}%
  \magictodo
  \the\magictoks}
\else
\long\def\later#1{#1}
\long\def\both#1{#1}
\fi
\def\magicappendix{\latertrue \the\magicAppendix}

\begin{document}

\title{Algorithms for Solving Rubik's Cubes}

\author{%
  Erik D. Demaine\inst{1}%
\and
  Martin L. Demaine\inst{1}%
\and
  Sarah Eisenstat\inst{1}%
\and \\
  Anna Lubiw\inst{2}%
\and
  Andrew Winslow\inst{3}%
}

\authorrunning{E. D. Demaine \and M. L. Demaine \and S. Eisenstat \and A. Lubiw \and A. Winslow}

\institute{
            MIT Computer Science and Artificial Intelligence Laboratory, \\
            Cambridge, MA 02139, USA,
            \protect\url{{edemaine,mdemaine,seisenst}@mit.edu}
            \and
            David R. Cheriton School of Computer Science, \\
            University of Waterloo, Waterloo, Ontario N2L 3G1, Canada,
            \protect\url{alubiw@uwaterloo.ca}
            \and
            Department of Computer Science, Tufts University, \\
            Medford, MA 02155, USA,
            \protect\url{awinslow@cs.tufts.edu}
}
 

\maketitle

\begin{abstract}
  The Rubik's Cube is perhaps the world's most famous and iconic puzzle,
  well-known to have a rich underlying mathematical structure (group theory).
  In this paper, we show that the Rubik's Cube also has a rich underlying
  algorithmic structure.  Specifically, we show that the $n \times n \times n$
  Rubik's Cube, as well as the $n \times n \times 1$ variant, has a ``God's
  Number'' (diameter of the configuration space) of $\Theta(n^2/\log n)$.
  The upper bound comes from effectively parallelizing standard
  $\Theta(n^2)$ solution algorithms, while the lower bound follows from
  a counting argument.  The upper bound gives an asymptotically optimal
  algorithm for solving a general Rubik's Cube in the worst case.
  Given a specific starting state, we show how to find the shortest solution
  in an $n \times O(1) \times O(1)$ Rubik's Cube.
  Finally, we show that finding this optimal solution becomes NP-hard
  in an $n \times n \times 1$ Rubik's Cube when the positions and colors
  of some of the cubies are ignored (not used in determining whether the
  cube is solved).

\end{abstract}



\section{Introduction}

A little over thirty years ago, Hungarian architecture professor
Ern\H{o} Rubik released his ``Magic Cube'' to the world.\footnote{Similar puzzles were invented around the same time in the United States~\cite{Gustafson-1963}\cite{Nichols-1972}, the United Kingdom~\cite{Fox-1974}, and Japan~\cite{Ishige-1976} but did not reach the same level of success.}
What we now all know as the \emph{Rubik's Cube} quickly became a sensation
\cite{Slocum-2009}.
It is the best-selling puzzle ever, at over 350 million units
\cite{Rubiks2010}.
It is a tribute to elegant design, being part of the permanent collection
of the Museum of Modern Art in New York \cite{MoMA}.
It is an icon for difficult puzzles---an intellectual Mount Everest.
It is the heart of World Cube Association's speed-cubing competitions,
whose current record holders can solve a cube in under 7 seconds
(or 31 seconds blindfold) \cite{WCA}.
It is the basis for cube art, a form of pop art made from many carefully
unsolved Rubik's Cubes.
(For example, the recent movie \emph{Exit Through the Gift Shop}
features the street cube artist known as Space Invader.)
It is the bane of many computers, which spent about 35 CPU years determining
in 2010 that the best algorithm to solve the worst configuration
requires exactly 20 moves---referred to as \emph{God's Number} \cite{cube20}.

To a mathematician, or a student taking abstract algebra,
the Rubik's Cube is a shining example of group theory.
The configurations of the Rubik's Cube, or equivalently the transformations
from one configuration to another, form a subgroup of a permutation group,
generated by the basic twist moves.
This perspective makes it easier to prove (and compute) that the configuration
space falls into two connected components, according to the parity of the
permutation on the cubies (the individual subcubes that make up the puzzle).
See \cite{Furst-Hopcroft-Luks-1980}
for how to compute the number of elements in the group generated
by the basic Rubik's Cube moves (or any set of permutations)
in polynomial time.

To a theoretical computer scientist, the Rubik's Cube and its many
generalizations suggest several natural open problems.
What are good algorithms for solving a given Rubik's Cube puzzle?
What is an optimal worst-case bound on the number of moves?
What is the complexity of optimizing the number of moves required for a
given starting configuration?
Although God's Number is known to be 20 for the $3 \times 3 \times 3$,
the optimal solution of each configuration in this constant-size puzzle
still has not been computed \cite{cube20}; even writing down the first move
in each solution would take about 8 exabytes (after factoring out symmetries).
While computing the exact behavior for larger cubes is out of the question,
how does the worst-case number of moves and complexity scale with the
side lengths of the cube?
In parallel with our work, these questions were recently posed by Andy Drucker
and Jeff Erickson \cite{cstheory}.
Scalability is important given the commercially available
$4 \times 4 \times 4$ Rubik's Revenge \cite{Sebesteny-1982};
$5 \times 5 \times 5$ Professor's Cube \cite{Krell-1986};
the $6 \times 6 \times 6$ and $7 \times 7 \times 7$ V-CUBEs \cite{VCUBE}
whose design enables cubes as large as $11
\times 11 \times 11$ according to Verdes's design patent \cite{Verdes-2007};
Leslie Le's 3D-printed $12 \times 12 \times 12$ \cite{12x12x12}; and
Oskar van Deventer's
$17 \times 17 \times 17$ Over the Top
and his $2 \times 2 \times 20$ Overlap Cube,
both available from 3D printer \emph{shapeways}~\cite{2x2x20}.

\paragraph{Diameter / God's Number.}
The diameter of the configuration space of a Rubik's Cube seems difficult to
capture using just group theory.  In general, a set of permutations (moves)
can generate a group with superpolynomial diameter \cite{Driscoll-Furst-1983}.
If we restrict each generator (move) to manipulate only $k$ elements,
then the diameter is $O(n^k)$ \cite{McKenzie-1984}, but this gives very
weak (superexponential) upper bounds for $n \times n \times n$
and $n \times n \times 1$ Rubik's Cubes.

Fortunately, we show that the general approach taken by folk algorithms
for solving Rubik's Cubes of various fixed sizes can be generalized to
perform a constant number of moves per cubie, for an upper bound of $O(n^2)$.
This result is essentially standard, but we take care to ensure that all
cases can indeed be handled.

Surprisingly, this bound is not optimal.  Each twist move in the
$n \times n \times n$ and $n \times n \times 1$ Rubik's Cubes
simultaneously transforms $n^{\Theta(1)}$ cubies (with the exponent depending
on the dimensions and whether a move transforms a plane or a half-space).
This property offers a form of parallelism for solving multiple cubies at once,
to the extent that multiple cubies want the same move to be applied at a
particular time.  We show that this parallelism can be exploited to reduce
the number of moves by a logarithmic factor, to $O(n^2 / \log n)$.
Furthermore, an easy counting argument shows an average-case lower bound of
$\Omega(n^2 / \log n)$, because the number of configurations is
$2^{\Theta(n^2)}$ and there are $O(n)$ possible moves from each
configuration.

Thus we settle the diameter of the $n \times n \times n$ and
$n \times n \times 1$ Rubik's Cubes, up to constant factors.
These results are described in Sections~\ref{nxnxn} and~\ref{nxn},
respectively.

\paragraph{$n^2-1$ puzzle.}
Another puzzle that can be described as a permutation group given by
generators corresponding to valid moves is the $n \times n$ generalization of
the classic Fifteen Puzzle.
This \emph{$n^2-1$ puzzle} also has polynomial diameter, though lacking any
form of parallelism, the diameter is simply $\Theta(n^3)$ \cite{Parberry-1995}.
Interestingly, however, computing the shortest solution
from a given configuration of the puzzle is NP-complete
\cite{Ratner-Warmuth-1990}.
More generally, given a set of generator permutations, it is PSPACE-complete
to find the shortest sequence of generators whose product is a given target
permutation \cite{Even-Goldreich-1981,Jerrum-1985}.
These papers mention the Rubik's Cube as motivation, but neither
address the natural question: is it NP-complete to solve a given
$n \times n \times n$ or $n \times n \times 1$ Rubik's Cube
using the fewest possible moves?
Although the $n \times n \times n$ problem was posed as early as
1984~\cite{Cook-1984,Ratner-Warmuth-1990}, 
both questions remain open~\cite{Kendall-2008}.
We give partial progress toward hardness, as well as a polynomial-time
exact algorithm for a particular generalization of the Rubik's Cube.


\paragraph{Optimization algorithms.}
We give one positive and one negative result about finding the shortest
solution from a given configuration of a generalized Rubik's Cube puzzle.
On the positive side, we show in Section~\ref{DP}
how to compute the exact optimum for
$n \times O(1) \times O(1)$ Rubik's Cubes.
Essentially, we prove structural results about how an optimal solution
decomposes into moves in the long dimension and the two short dimensions,
and use this structure to obtain a dynamic program.
This result may prove useful for optimally solving configurations of
Oskar van Deventer's $2 \times 2 \times 20$ Overlap Cube \cite{2x2x20},
but it does not apply to the $3 \times 3 \times 3$ Rubik's Cube
because we need $n$ to be distinct from the other two side lengths.
On the negative side, we prove in Section~\ref{hardness}
that it is NP-hard to find an optimal solution to a subset of cubies
in an $n \times n \times 1$ Rubik's Cube.  Phrased differently, optimally
solving a given $n \times n \times 1$ Rubik's Cube configuration is
NP-hard when the colors and positions of some cubies are ignored
(i.e., they are not considered in determining whether the cube is solved).

%
%
%


\section{Common Definitions}

We begin with some terminology.
An $\ell \times m \times n$ Rubik's Cube
is composed of $\ell m n$ \concept{cubies},
each of which has some position $(x, y, z)$,
where $x \in \group{\ell}$,
$y \in \group{m}$,
and $z \in \group{n}$.
Each cubie also has an orientation.
Each cubie in a Rubik's Cube
has a color on each visible face.
There are six colors in total.
We say that a Rubik's Cube is \concept{solved}
when each face of the cube
is the same color,
unique for each face.

A \concept{slice} of a Rubik's Cube
is a set of cubies that match in one coordinate
(e.g. all of the cubies such that $y = 1$).
A legal move on a Rubik's Cube
involves rotating one slice
around its perpendicular\footnote{While other reasonable definitions of a legal move exist (e.g. rotating a set of contiguous parallel slices), this move definition most closely matches the definition used in popular move notations.}.
In order to preserve the shape of the cube,
there are restrictions on how much
the slice can be rotated.
If the slice to be rotated is a square,
then the slice can be rotated
$90 \degree$ in either direction.
Otherwise,
the slice can only be rotated
by $180 \degree$.
Finally,
note that if one dimension of the cube
has length $1$,
we disallow rotations of the only slice in that dimension.
For example, we cannot rotate the slice $z = 0$
in the $n \times n \times 1$ cube.

A \concept{configuration} of a Rubik's Cube
is a mapping from each visible face of each cubie
to a color.
A \concept{reachable configuration} of a Rubik's Cube
is a configuration which can be reached
from a solved Rubik's Cube
via a sequence of legal moves.

For each of the Rubik's Cube variants we consider,
we will define the contents of a \concept{\CCS{}}.
The cubies which belong in this \CCS{}
depend on the problem we are working on;
however, they do share some key properties:
\begin{enumerate}

\item Each \CCS{} consists of a constant number of cubies.

\item No sequence of legal moves
can cause any cubie
to move from one \CCS{}
into another.

\end{enumerate}
Each \CCS{} has a \concept{\CCSconfig{}}
mapping from each visible face of the \CCS{}
to its color.
Because the number of cubies in a \CCS{} is constant,
the number of possible \CCSconfigs{}
is also constant.

We say that a move \concept{affects} a \CCS{}
if the move causes at least one cubie in the \CCS{}
to change places.
Similarly,
we say that a sequence of moves
affects a \CCS{}
if  at least one cubie in the \CCS{}
will change position or orientation
after the sequence of moves
has been performed.


\section{Diameter of $n \times n \times 1$ Rubik's Cube}
\label{nxn}

\newcommand{\up}{red}
\newcommand{\down}{blue}
\newcommand{\Up}{Red}
\newcommand{\Down}{Blue}

When considering
an $n \times n \times 1$ Rubik's Cube
we omit the third coordinate of a cubie,
which by necessity must be $0$.
For simplicity,
we restrict the set of solutions
to those configurations where
the top of the cube is \up{}
and the bottom of the cube is \down{}.

Consider the set of locations that
a cubie at position $(x, y)$ can reach.
If we flip column $x$, the cubie will move
to position $(x, n - y - 1)$.
If we instead flip row $y$, the cubie will move
to position $(n - x - 1, y)$.
Consequently, there are at most four reachable
locations for a cubie that starts at $(x, y)$:
$(x, y)$, $(x, n - y - 1)$, $(n - x - 1, y)$, and $(n - x - 1, n - y - 1)$.
We call this set of locations
the \CCS{} $(x, y)$.

\later{
If a cubie is in the first or last row or column, we call it an \emph{edge cubie}.
If a cluster contains an edge cubie,
then we call it an \emph{edge cluster},
because all of its cubies are edge cubies.
The special edge cluster
which contains the cubie in the first row and first column
is called the \emph{corner cluster},
and its cubies are \emph{corner cubies}.
If $n$ is odd, then there is a cluster with one cubie
which is in both the median row and the median column.
We will call this cubie
the \emph{center cubie} or \emph{center cluster}.
In addition, if $n$ is odd then
there are also clusters with two cubies
found in the median row or column.
We call the clusters \emph{cross clusters}
and the cubies in them \emph{cross cubies}.
}

We begin by showing that
for any reachable cluster configuration,
there exists a sequence of moves
of constant length
which can be used to solve that cluster
without affecting any other clusters.

In the remainder of Section~\ref{nxn}, we use the notation $H_1, H_2$ and
$V_1, V_2$ to denote the two rows and columns containing cubies from a single
cubie cluster.
We also use the same symbols to denote single moves affecting these rows and columns.
Recall that for a $n \times n \times 1$ cube there is only
one valid non-identity operation on a row or column: rotation by
$180^{\circ}$.  In the special cases of cross and center cubie clusters,
we denote the single row or column containing the cluster by $H$ or $V$,
respectively. 

\iffull
We begin by proving the following lemma
about cubie cluster configurations.
\fi
\later{
\begin{lemma}
\label{lemma:nxnx1-group-config}
In a solvable $n \times n \times 1$ Rubik's Cube,
the colors on the top faces of the cubies in a 
cubie cluster can only be in the six configurations
in Fig.~\ref{figure:moves}.
\end{lemma}
}

\later{
\begin{proof}
Consider what happens to a \CCS{}
when a move is performed.
If the move does not affect the \CCS{},
then its cubie configuration will not change.
Otherwise,
the move will swap two cubies
in the same row or column
while reversing the color of each cubie.
If both cubies are the same color,
then both cubies will become the other color.
In other words, if one cubie is \up{}
and the other is \down{},
then the color configuration will not change.

Figure \ref{figure:all-red} shows the solved configuration,
which is quite obviously reachable.
The four moves which affect this \CCS{}
result in configurations
\ref{figure:blue-left}, \ref{figure:blue-right},
\ref{figure:blue-up}, and \ref{figure:blue-down}.
Consider how the four possible moves affect
each of configurations
\ref{figure:blue-left}, \ref{figure:blue-right},
\ref{figure:blue-up}, and \ref{figure:blue-down}.
For each configuration,
two of the four possible moves
involve one \up{} cubie and one \down{} cubie,
and therefore do not affect the colors.
In addition,
one move for each configuration
is the inverse of the move used to reach that configuration,
and therefore leads back to configuration \ref{figure:all-red}.
The final move for each configuration
results in configuration \ref{figure:all-blue},
thus completing the set of reachable configurations.
\qed
\end{proof}
}

\ifabstract
\else
\setlength{\unitlength}{0.01in}
\newcommand{\figuremoves}[1]{
	\hspace{0.05in}
	\begin{picture}(120, 100)
	\put(10, -10){\includegraphics[width=100\unitlength]{#1}}
	\put(3, 50){$H_1$}
	\put(3, 27){$H_2$}
	\put(33, 85){$V_1$}
	\put(76, 85){$V_2$}
	\end{picture}	
	\hspace{0.05in}
}
\fi

\begin{figure}
\centering
\subfigure[Solved.]{
	\figuremoves{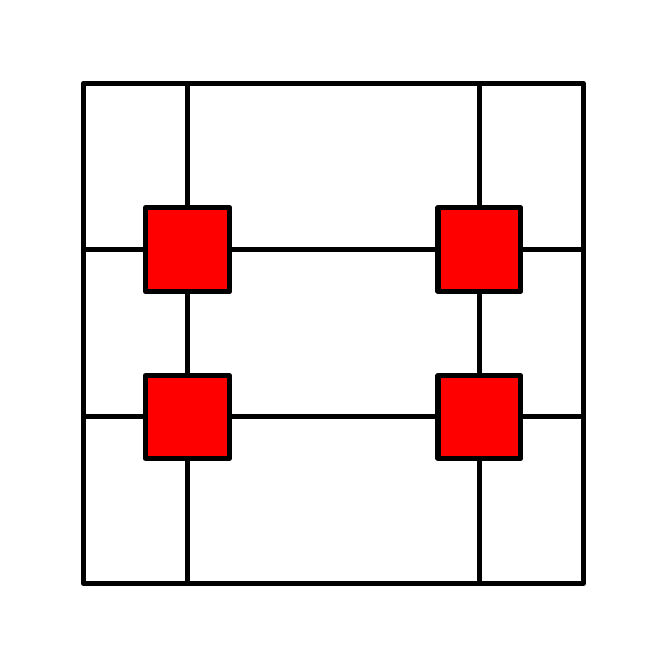}
	\label{figure:all-red}
}
\subfigure[$V_1, H_1, V_1, H_1.$]{
	\figuremoves{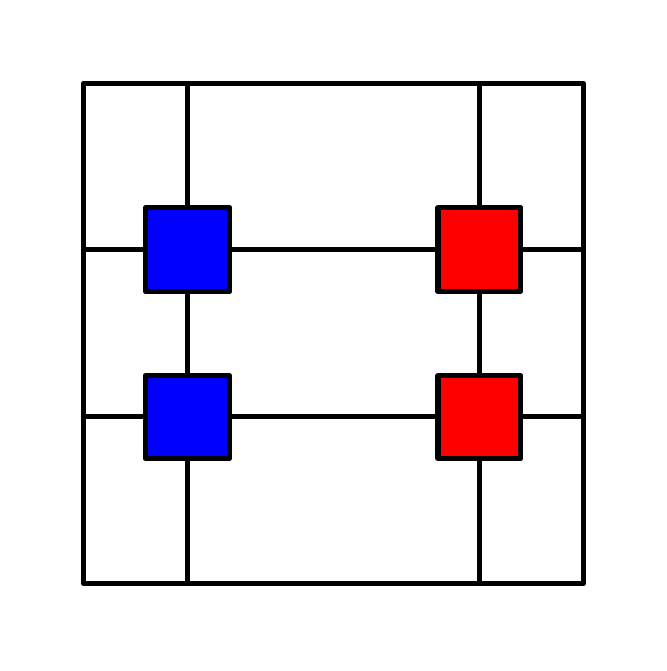}
	\label{figure:blue-left}
}
\subfigure[$V_2, H_1, V_2, H_1.$]{
	\figuremoves{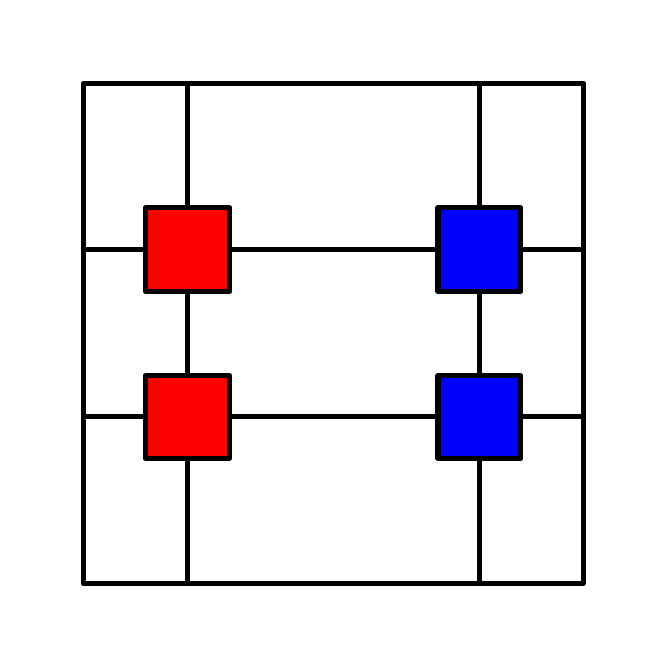}
	\label{figure:blue-right}
} \\
\subfigure[$H_1, V_1, H_1, V_1.$]{
	\figuremoves{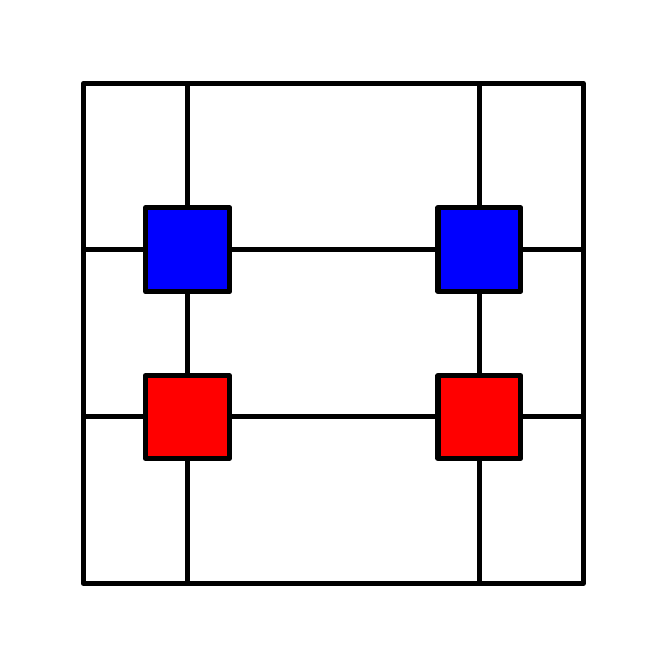}
	\label{figure:blue-up}
}
\subfigure[$H_2, V_1, H_2, V_1.$]{
	\figuremoves{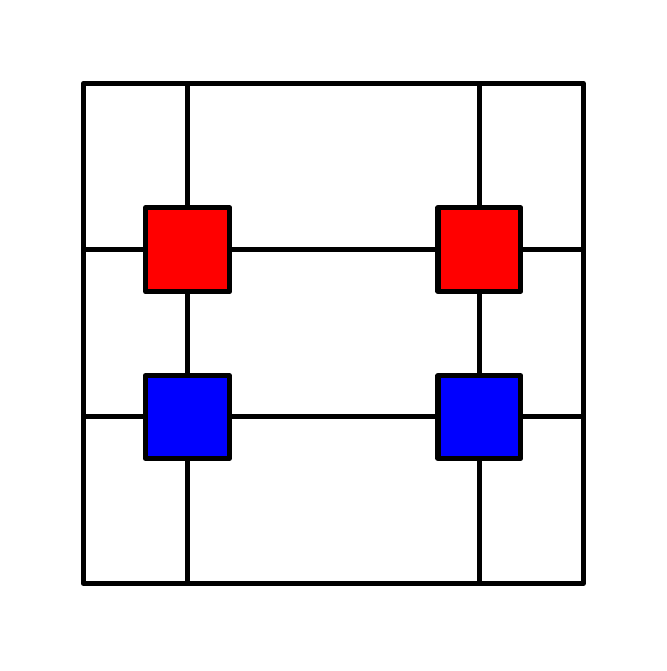}
	\label{figure:blue-down}
}
\subfigure[$H_1, H_2, V_1, H_1, H_2, V_1.$]{
	\figuremoves{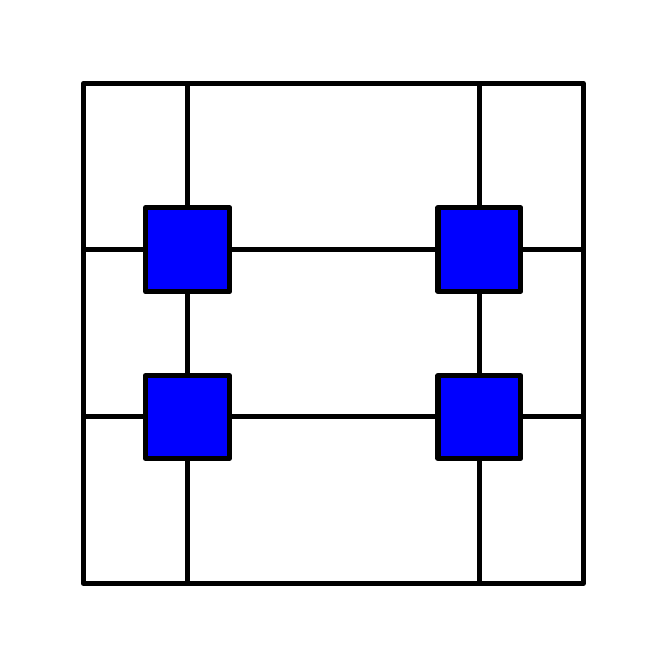}
	\label{figure:all-blue}
}
\caption{The reachable \CCSconfigs{}
and the move sequences to solve them.}
\label{figure:moves}
\end{figure}

\iffull
Because of these limitations on cubie cluster configurations,
we can prove the following lemma.
\fi

\later{
\begin{lemma}
\label{lemma:nxnx1-group-reachable}
All six \CCSconfigs{} from Lemma \ref{lemma:nxnx1-group-config}
can be solved using a sequence of at most six moves
that does not affect the position
of any cubies not in the \CCS{}.  (See Fig.~\ref{figure:moves}.)
\end{lemma}
}

\later{
\begin{proof}
The correct move sequence for each configuration
from Lemma \ref{lemma:nxnx1-group-config} is given in Fig.~\ref{figure:moves}.
The fact that we always use each move twice
ensures that all other clusters
will not be affected by the move sequence.
The correctness of these sequences can be verified by the reader.
\qed
\end{proof}
}

\iffull
In order to handle
the edge, corner, and cross clusters,
we need a more complicated sequence of moves.
These clusters cannot always be solved
without affecting any other cubes.
So rather than show that we can solve each cluster individually,
we show that we can solve all such clusters together.
\fi

\later{
\begin{lemma}
Given a solvable configuration
of an $n \times n \times 1$ Rubik's Cube,
there exists a sequence of moves of length $O(n)$
which can be used to solve the edge clusters
and cross clusters.
\end{lemma}
}

\later{
\begin{proof}
We begin by solving the corner cluster,
and, if $n$ is odd,
the center cluster and the two edge cross clusters.
These four clusters combined have only $O(1)$
reachable configurations,
and so all four can be obviously be fixed in $O(1)$ moves,
disregarding the effect that these moves might have
on any other clusters.
Next we solve the other edge clusters.
Our goal is to solve each cluster
without affecting any of the clusters
we have previously solved.

Without loss of generality,
say that we are trying to solve
an edge cluster with coordinates $(x, 0)$.
We begin by using the move sequences
from Lemma \ref{lemma:nxnx1-group-reachable}
to make sure that the cluster
has all four red stickers visible.
Then it will be in one of the states
depicted in Fig.~\ref{fig:nxnx1-edge-configs}.
To solve the cluster,
we can use the move sequences given
in Fig.~\ref{fig:nxnx1-edge-configs}.
Although the given move sequences
are not guaranteed to apply the identity permutation
to all other clusters,
they do have the property
that any horizontal move
will be performed an even number of times.
Hence, this move sequence will apply
the identity permutation to all other edge clusters.
In addition, none of the move sequences affect the center cubie.

\setlength{\unitlength}{0.008in}
\newcommand{\figureedges}[1]{
	\begin{picture}(120, 115)
	\put(10, 0){\includegraphics[width=100\unitlength]{#1}}
	\put(0, 78){$H_1$}
	\put(0, 17){$H_2$}
	\put(35, 103){$V_1$}
	\put(72, 103){$V_2$}
	\end{picture}
}

\begin{figure}
\centering
\subfigure[Solved.]{
	\figureedges{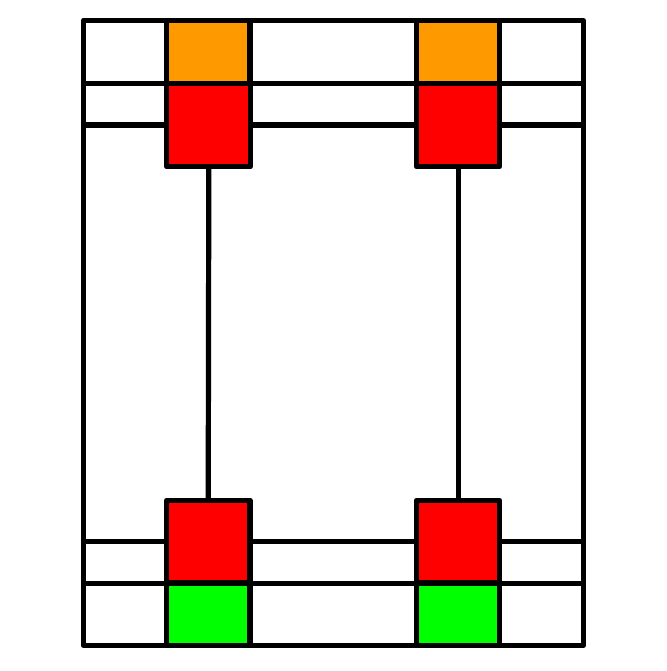}
}
\subfigure[$H_1, V_1, H_1.$]{
	\figureedges{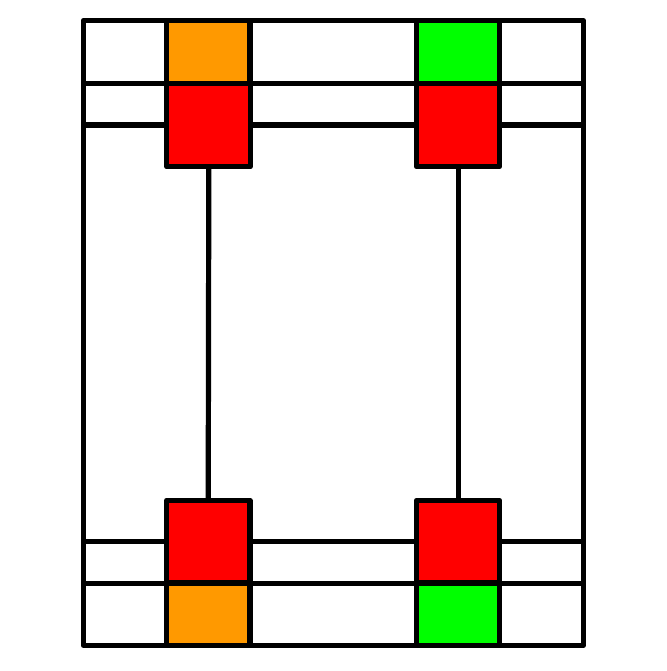}
} 
\subfigure[$H_1, V_2, H_1.$]{
	\figureedges{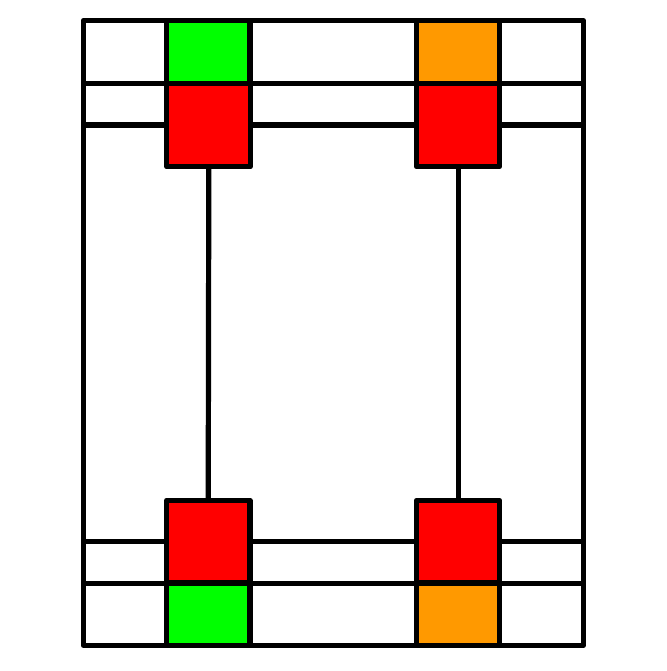}
}
\subfigure[$H_1, V_1, V_2, H_1.$]{
	\figureedges{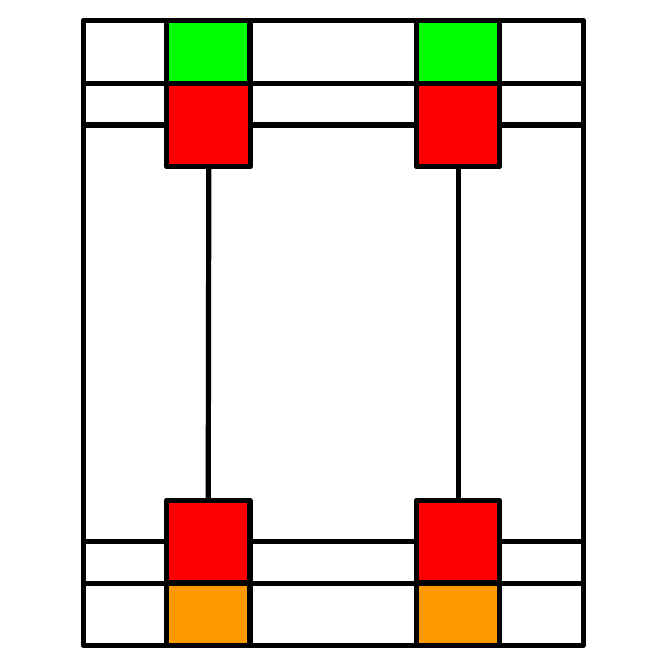}
}
\caption{
When all four red stickers are facing forwards,
these are the possible configurations for an edge cluster,
and the move sequences used to solve them.
}
\label{fig:nxnx1-edge-configs}
\end{figure}

Once all the edge clusters have been solved,
we want to solve the cross clusters.
We know that the center cluster has already been solved.
We also know that there are $O(n)$ cross clusters,
so if we can solve each cross cluster in $O(1)$ moves
without affecting the rest of the clusters,
then we will have solved the edge and cross clusters
in a total of $O(n)$ moves.
Without loss of generality,
say that we are trying to solve the cross cluster
$((n - 1) / 2, y)$.
The four possible states for a cross cluster
are depicted in Fig.~\ref{figure:cross-states}.

\setlength{\unitlength}{0.008in}
\newcommand{\figurecross}[1]{
	\begin{picture}(120, 100)
	\put(10, -10){\includegraphics[width=100\unitlength]{#1}}
	\put(0, 56){$H_1$}
	\put(0, 18){$H_2$}
	\put(55, 85){$V$}
	\end{picture}
}
\begin{figure}
\centering
\subfigure[Solved.]{
	\figurecross{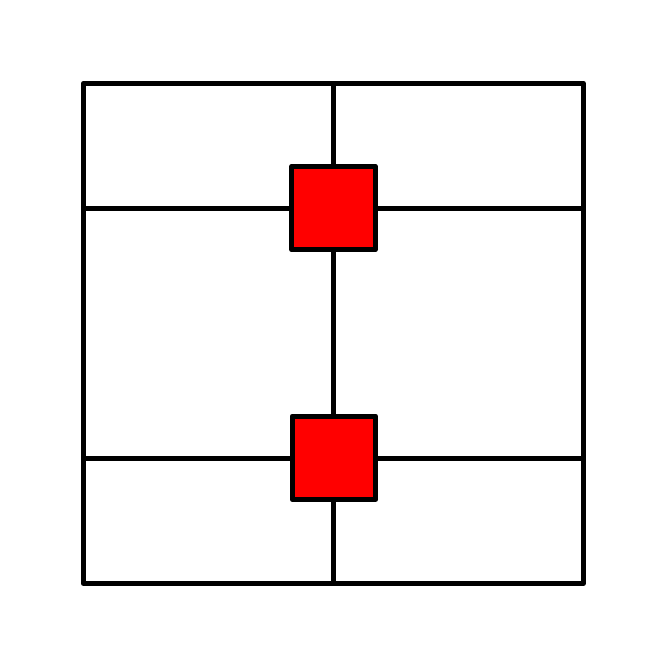}
	\label{figure:cross-states-a}
}
\subfigure[n/a.]{
	\figurecross{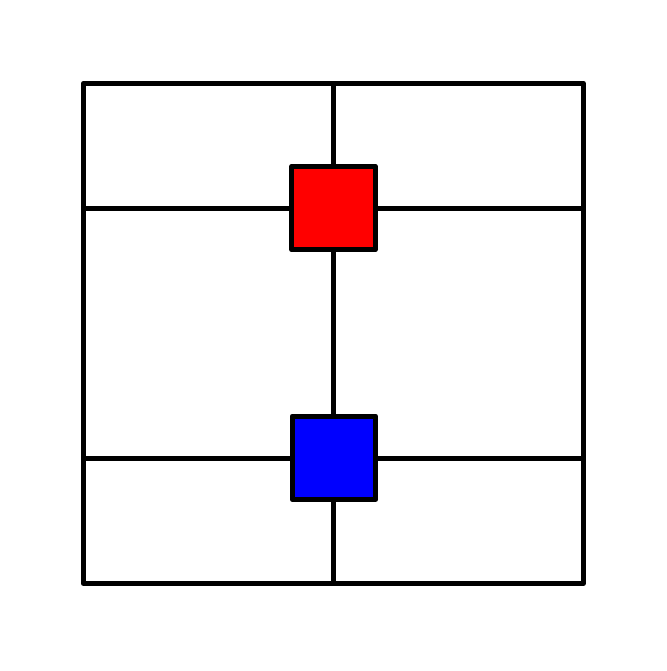}
	\label{figure:cross-states-b}
}
\subfigure[n/a.]{
	\figurecross{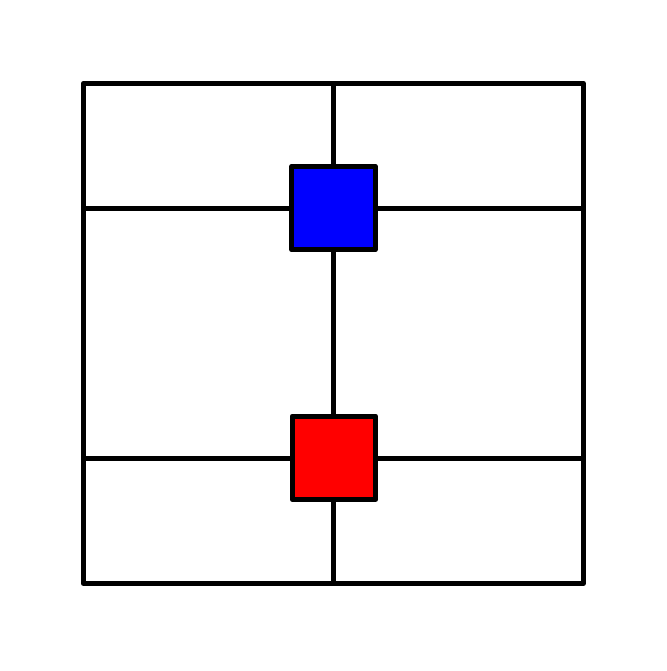}
	\label{figure:cross-states-c}
}
\subfigure[$H_1, V, H_1, V.$]{
	\figurecross{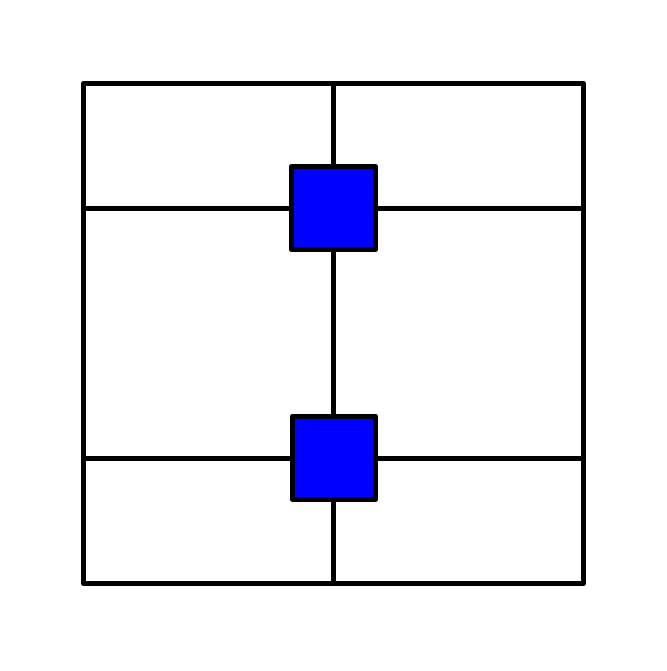}
	\label{figure:cross-states-d}
}
\caption{
The four reachable configurations for a cross cluster.
Two of them can be solved without affecting
the rest of the clusters.
}
\label{figure:cross-states}
\end{figure}

However, because we have already solved
all of the edge clusters,
we know that the set of possible configurations
for our cross cluster is more restricted.
Both horizontal moves
affecting our cross cluster
will cause one of the two cubies
to change color.
No matter what state the cross cluster is in,
the vertical move cannot change
the color of only one cubie.
Therefore, if the cross cluster
is in the configuration depicted
in Fig.~\ref{figure:cross-states-b}
or the configuration depicted in
Fig.~\ref{figure:cross-states-c},
then the rest of the solution must perform
either the move $H_1$ or the move $H_2$
an odd number of times.

Consider the effect of this on the edge cluster
affected by $H_1$ and $H_2$.
Each move in the edge cluster causes a swap of two cubies.
If an order is placed on the cubies in the cluster,
each swap is a permutation of this order,
and the set of permutations constructible using swaps
is equivalent to the permutation group on 4 elements, $S_{4}$.
By permutation group theory,
if a particular cluster configuration can be solved
using an even number of swaps,
then \emph{any} solution for that cluster configuration
has an even number of swaps.
We know that the edge cluster is already solved,
so it can be solved using an even number of swaps.
So if the rest of the solution
contains an odd number of moves $H_1$ and $H_2$,
then it must also contain an odd number
of edge moves $V_1$ and $V_2$.

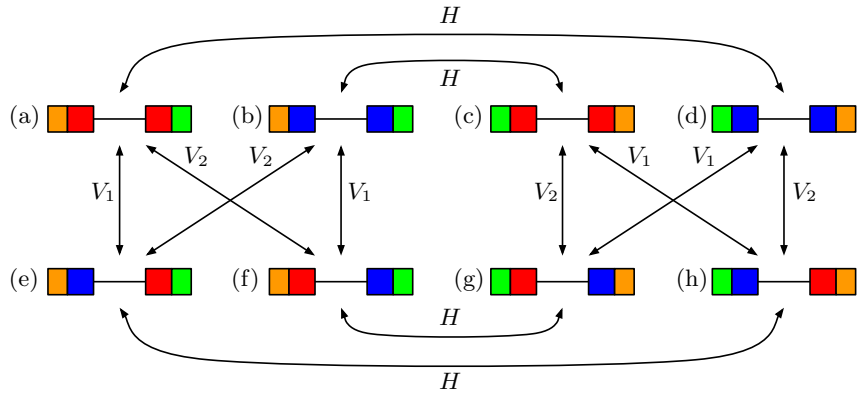
\begin{figure}
\centering
\setlength{\unitlength}{0.0075in}
\begin{picture}(600, 270)
\put(4, 0){\includegraphics[width=4.5in]{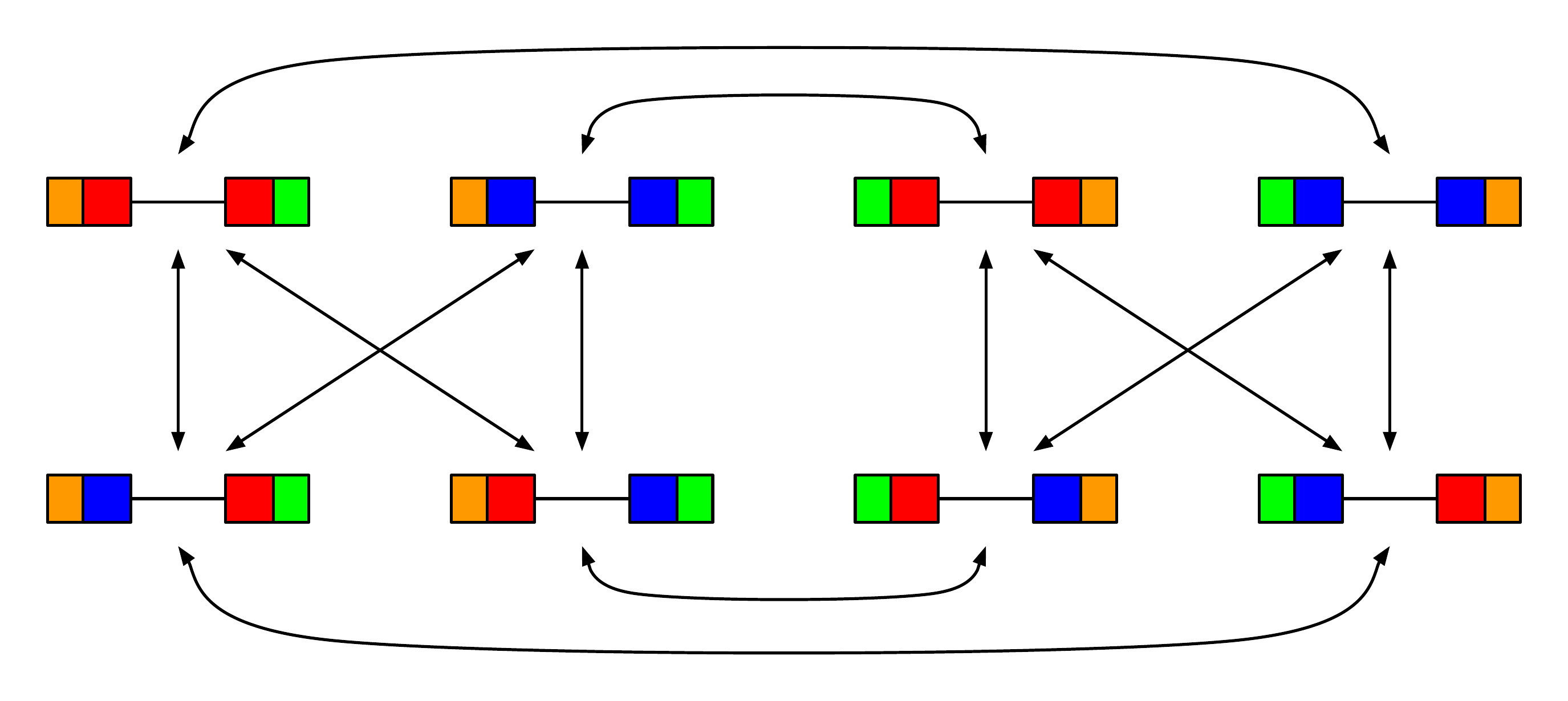}}
\put(-5, 188){(a)}
\put(150, 188){(b)}
\put(305, 188){(c)}
\put(460, 188){(d)}
\put(-5, 74){(e)}
\put(150, 74){(f)}
\put(305, 74){(g)}
\put(460, 74){(h)}
\put(295, 2){$H$}
\put(295, 47){$H$}
\put(295, 214){$H$}
\put(295, 258){$H$}
\put(52, 135){$V_1$}
\put(231, 135){$V_1$}
\put(361, 135){$V_2$}
\put(541, 135){$V_2$}
\put(117, 160){$V_2$}
\put(162, 160){$V_2$}
\put(426, 160){$V_1$}
\put(471, 160){$V_1$}
\end{picture}
\caption{
The possible configurations
of an edge cross cluster.
}
\label{figure:edge-cross-configs}
\end{figure}
Now consider the effect of an odd number of edge moves
on the affected edge cross cluster.
Figure~\ref{figure:edge-cross-configs}
gives the configuration space for that cluster.
We know that it is currently in the solved state labelled (a).
Any sequence of moves
which contains an odd number of edge moves $V_1$ and $V_2$
will cause the edge cross cluster
to leave the solved state.
So the rest of the solution cannot contain
an odd number of edge moves $V_1$ and $V_2$.
This means that the rest of the solution
cannot contain an odd number of horizontal moves
$H_1$ and $H_2$
affecting a single cross cluster.
So every cross cluster must be in one of the states
depicted in Figs.~\ref{figure:cross-states-a} and \ref{figure:cross-states-d},
each of which can be solved without affecting any other clusters
using the sequence of moves listed for each state.
\qed
\end{proof}
}

\subsection{$n \times n \times 1$ Upper Bound}

There are $n^2$ clusters in the $n \times n \times 1$ Rubik's Cube.
If we use the move sequences given in Fig.~\ref{figure:moves}
to solve each cluster individually,
we have a sequence of $O(n^2)$ moves for solving the entire cube.
In this section,
we take this sequence of moves
and take advantage of parallelism
to get a solution with $O(n^2 / \log n)$ moves.

Say that we are given columns $X$ and rows $Y$
such that all of the cubie clusters $(x, y) \in X \times Y$
are in the configuration depicted in Fig.~\ref{figure:blue-left}.
If we attempted to solve each of these clusters individually,
the number of moves required would be $O(|X| \cdot |Y|)$.

Consider instead what would happen
if we first flipped all of the columns $x \in X$,
then flipped all of the rows $y \in Y$,
then flipped all of the columns $x \in X$ again,
and finally flipped all of the rows $y \in Y$ again.
What would be the effect of this move sequence
on a particular $(x^*, y^*) \in X \times Y$?
The only moves affecting that cluster
are the column moves $x^*$ and $(n - 1 - x^*)$
and the row moves $y^*$ and $(n - 1 - y^*)$.
So the subsequence of moves affecting $(x^*, y^*)$
would consist of
the column move $x^*$,
followed by the row move $y^*$,
followed by the column move $x^*$ again,
and finally the row move $y^*$ again.
Those four moves
are exactly the moves needed to solve that cluster.

\iffull
This idea allows us to parallelize
the solutions for multiple clusters,
resulting in the following lemma.
\fi

\later{
\begin{lemma}
\label{lemma:nxnx1-bulk-same}
Given an $n \times n \times 1$ Rubik's Cube configuration
and sets $X, Y \subseteq \group{\floor{n / 2}}$,
if all \CCSs{} $(x, y) \in X \times Y$
are in the same \CCSconfigs{},
then all \CCSs{} $(x, y) \in X \times Y$
can be solved in $O(|X| + |Y|)$ moves
without affecting the rest of the cubies.
\end{lemma}
}

\later{
\begin{proof}
By Lemma \ref{lemma:nxnx1-group-reachable},
we know that there exists a sequence of moves
of length at most six
which will solve the configuration
of a single \CCS{} $(x, y) \in X \times Y$.
We write this sequence of moves
in a general form as in Lemma \ref{lemma:nxnx1-group-reachable}.
We then replace each move $V_1$
with a sequence of moves
that flips each column $x \in X$.
Similarly, we replace each move $V_2$
with a sequence of moves
that flips each column $n - x - 1$, where $x \in X$.
We perform similar substitutions
for $H_1$ and $H_2$,
using the rows $y$ and $n - y - 1$ instead.
Because the original sequence of moves
had length at most six,
it is easy to see that the length
of the new move sequence
is $O(|X| + |Y|)$.

\begin{figure}[h]
\centering
\subfigure[]{
	\includegraphics[width=0.17\textwidth, trim=0.8cm 0.8cm 0.8cm 0.8cm, clip]{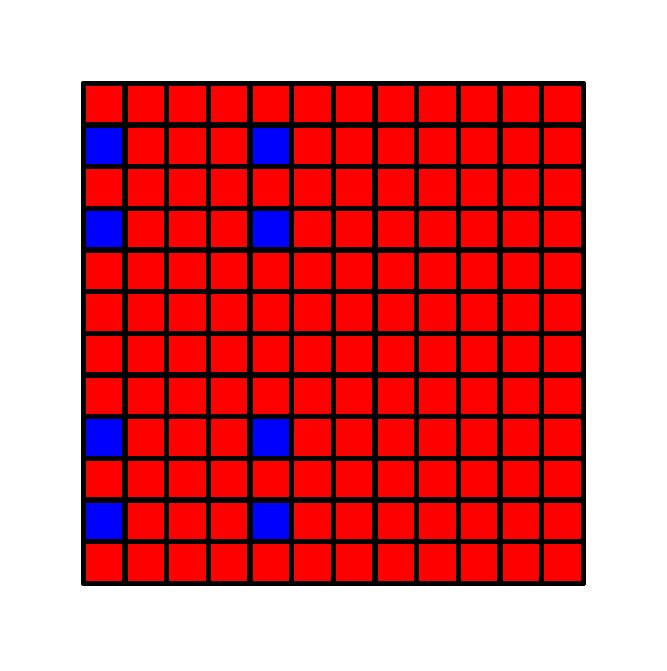}
}
\subfigure[]{
	\includegraphics[width=0.17\textwidth, trim=0.8cm 0.8cm 0.8cm 0.8cm, clip]{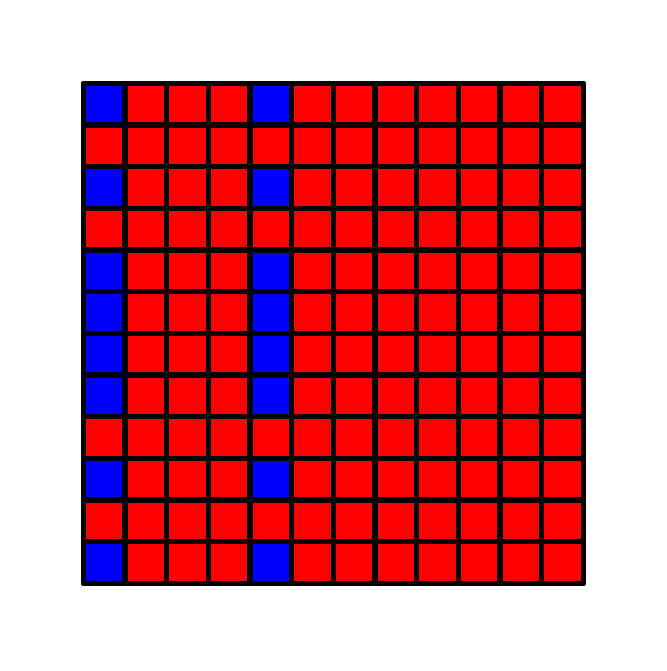}
}
\subfigure[]{
	\includegraphics[width=0.17\textwidth, trim=0.8cm 0.8cm 0.8cm 0.8cm, clip]{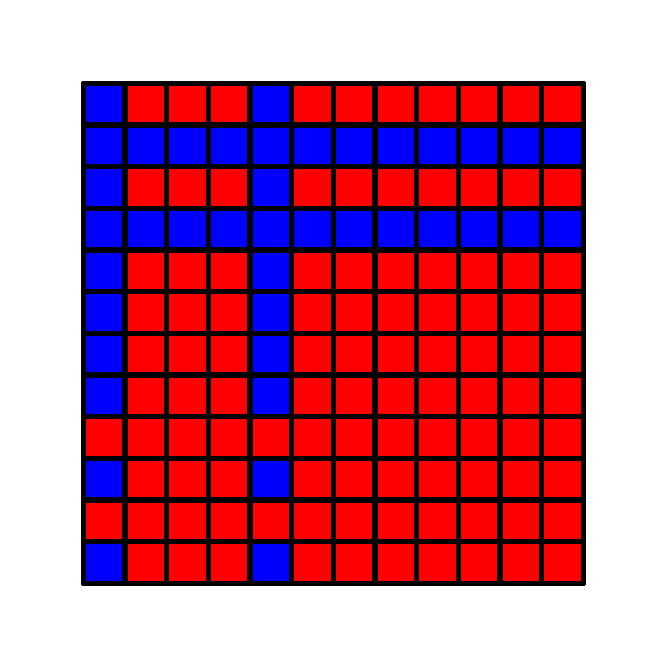}
}
\subfigure[]{
	\includegraphics[width=0.17\textwidth, trim=0.8cm 0.8cm 0.8cm 0.8cm, clip]{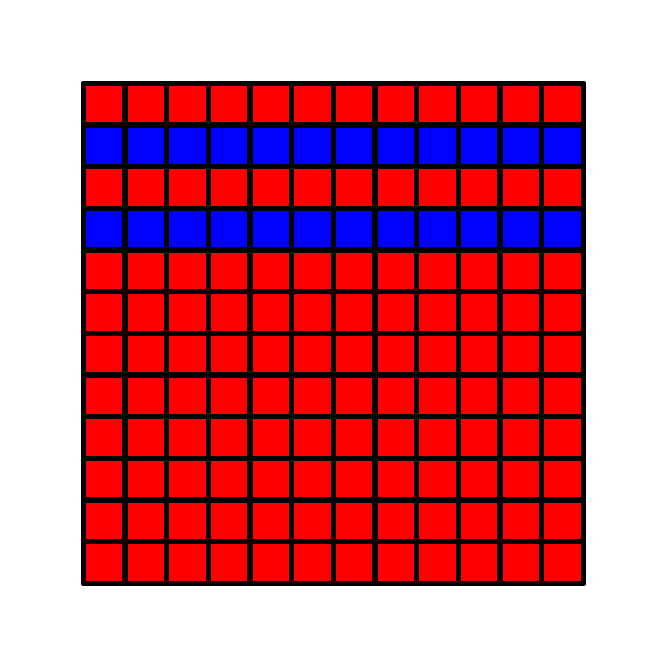}
}
\subfigure[]{
	\includegraphics[width=0.17\textwidth, trim=0.8cm 0.8cm 0.8cm 0.8cm, clip]{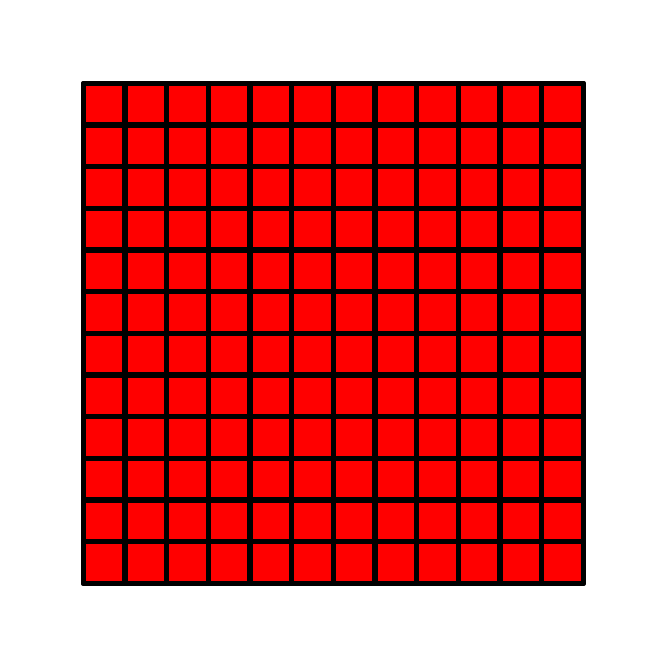}
}
\caption{Solving four \CCSs{} at the same time.}
\end{figure}

We claim that this new sequence of moves
will solve all \CCSs{} $(x, y) \in X \times Y$,
and that no other \CCSs{} will be affected.
To see that this is true,
consider how some \CCS{}
$(x^*, y^*) \in \group{\floor{n / 2}} \times \group{\floor{n / 2}}$
will be affected by the move sequence.
If $x^* \notin X$ and $y^* \notin Y$,
then none of the moves in the sequence
will affect the position of any cubies in the \CCS{},
and therefore the \CCS{} will be unaffected.
If $x^* \in X$ and $y^* \in Y$,
then the subsequence of moves which affect this \CCS{}
will be exactly the sequence of moves
necessary to solve this \CCS{}.
Otherwise,
either $x^* \in X$ or $y^* \in Y$,
but not both.
Therefore,
the subsequence of moves
which affect this \CCS{}
will be either all vertical,
or all horizontal,
and each move will occur exactly twice.
This means that the subsequence of moves
which affect this \CCS{}
will apply the identity permutation
to this \CCS{}.
\qed
\end{proof}
}

\ifabstract
\else
This technique allows us to solve
all cubie clusters $(x, y) \in X \times Y$
using only $O(|X| + |Y|)$ moves,
if each one of those clusters
is in the same configuration.
\fi
Our goal is to use this technique
for a related problem:
solving all of the cubie clusters $(x, y) \in X \times Y$
that are in a particular cluster configuration $c$,
leaving the rest of the clusters alone.
\ifabstract
\else
To that end, we divide up the columns $X$
according to the pattern of rows that are in configuration $c$,
and solve each subset of the rows
using the technique of Lemma~\ref{lemma:nxnx1-bulk-same}.
More formally:
\fi

\later{
\begin{lemma}
\label{lemma:nxnx1-bulk-subsets}
Suppose we are given
an $n \times n \times 1$ Rubik's Cube configuration,
a \CCSconfig{} $c$,
and sets $X, Y \subseteq \group{\floor{n / 2}}$
such that $|X| = \ell$.
Then all \CCSs{} $(x, y) \in X \times Y$
that are in configuration $c$
can be solved in $O(\ell 2^\ell + |Y|)$ moves
without affecting the rest of the cubies.
\end{lemma}
}

\later{
\begin{proof}
For each $y \in Y$,
we define $S_y = \{ x \in X \mid$ the \CCS{} $(x, y)$ is in configuration $c \}$.
For each $S \subseteq X$,
we let $Y_S = \{ y \in Y \mid S_y = S \}$.
There are $2^\ell$ possible values for $S$.
For each $Y_S$,
we use the sequence of moves
which is guaranteed to exist by Lemma \ref{lemma:nxnx1-bulk-same}
to solve all $(x, y) \in S \times Y_S$.
This sequence of moves has length $O(|S| + |Y_S|) = O(\ell + |Y_S|)$.
When we sum the lengths up
for all $Y_S$,
we find that the number of moves is bounded by
\begin{align*}
O\left(\ell \cdot 2^\ell + \sum_S |Y_S|\right) = O\left(\ell \cdot 2^\ell + |Y|\right).
\end{align*}
\qed
\end{proof}
}

\ifabstract
\else

Unfortunately, this result is not cost-effective.
We need to make sure that $\ell$ is small enough
to prevent an exponential blowup in the cost of solving the Rubik's Cube.
To that end, we divide up the columns
into small groups to get the following result:

\fi

\later{
\begin{lemma}
\label{lemma:nxnx1-bulk-et-al}
Suppose we are given
an $n \times n \times 1$ Rubik's Cube configuration,
a \CCSconfig{} $c$,
and sets $X, Y \subseteq \group{\floor{n / 2}}$.
Then all \CCSs{} $(x, y) \in X \times Y$
that are in configuration $c$
can be solved in $O(|X| \cdot |Y| / \log |Y|)$ moves
without affecting the rest of the cubies.
\end{lemma}
}

\later{
\begin{proof}
Let $\ell = \frac{1}{2} \log_2 |Y|$,
so that $2^\ell = \sqrt{|Y|}$.
Let $k = \lceil |X| / \ell \rceil$.
Partition the set $X$ into
a series of sets $X_1, \ldots, X_{k}$
each of which has size $\le \ell$.
For each $X_i$,
we solve the \CCSs{} in $X_i \times Y$
using the sequence of moves that is guaranteed to exist
by Lemma \ref{lemma:nxnx1-bulk-subsets}.
The number of moves required to solve
a single $X_i$ is
\begin{align*}
O\left(\ell 2^\ell + |Y| \right) = O\left(\left(\frac{1}{2} \log_2 |Y|\right) \sqrt{|Y|} + |Y| \right) = O(|Y|).
\end{align*}
Therefore, if we wish to perform this for $k$ sets,
the total number of moves becomes:
\begin{align*}
O\left(k \cdot |Y|\right) = O\left( \frac{|X|}{\frac{1}{2} \log_2 |Y|} \cdot |Y|\right) = O\left( \frac{|X| \cdot |Y|}{ \log |Y|} \right)
\end{align*}
\qed
\end{proof}
}

\iffull
As a result of this parallelization,
we get the following upper bound
on the diameter of the configuration space:
\fi

\both{
\begin{theorem}
\label{theorem:nxnx1-upper-bound}
Given an $n \times n \times 1$ Rubik's Cube configuration,
all \CCSs{} can be solved in $O(n^2 / \log n)$ moves.
\end{theorem}
}

\later{
\begin{proof}
In order to solve the Rubik's Cube,
we must solve all \CCSs{}
$(x, y) \in \group{\floor{n / 2}} \times \group{\floor{n / 2}}$.
To do so,
we loop through the six possible \CCSconfigs{},
then use the sequence of moves guaranteed to exist
by Lemma \ref{lemma:nxnx1-bulk-et-al}
to solve all of the \CCSs{} which are in a particular configuration.
For a single configuration,
the number of moves that this generates is
\begin{align*}
O\left( \frac{\floor{n/2} \cdot \floor{n / 2}}{\log(\floor{n / 2})}\right) = O\left(\frac{n^2}{\log n}\right).
\end{align*}
When we add this cost up
for the six different configurations,
the total number of moves is
$6 \cdot O(n^2 / \log n) = O(n^2 / \log n)$.
\qed
\end{proof}
}

\subsection{$n \times n \times 1$ Lower Bound}

\ifabstract
\else
In this section,
we establish the matching lower bound:
\fi

\both{
\begin{theorem}
\label{theorem:nxnx1-lower-bound}
Some configurations
of an $n \times n \times 1$ Rubik's Cube
are $\Omega(n^2 / \log n)$ moves away from being solved.
\end{theorem}
}

\later{
\begin{proof}
Lemma \ref{lemma:nxnx1-group-reachable} shows that
for every possible configuration of a \CCS{},
there exists a sequence of moves to solve the \CCS{}
while leaving the rest of the cubies in the same location.
Hence,
the inverse of such a sequence
will transform a solved \CCS{}
to an arbitrary \CCSconfig{}
without affecting any other cubies.
Not counting the edge cubies and the cross cubies,
there are $\left(\lfloor n / 2 \rfloor - 1\right)^2$ \CCSs{},
each of which can be independently placed
into one of six different configurations.
This means that there are at least
$6^{\left(\lfloor n / 2 \rfloor - 1\right)^2}$ reachable configurations.

There are $2n$ possible moves.
Therefore,
the total number of states reachable
using at most $k$ moves
is at most
\begin{align*}
\frac{(2n)^{k + 1} - 1}{2n - 1} \le (2n)^{k + 1}.
\end{align*}
Therefore,
if $k$ is the number of moves
necessary to reach all states,
it must have the property that:
\begin{align*}
6^{\left(\lfloor n / 2 \rfloor - 1\right)^2} &\le (2n)^{k + 1}, \\
\left(\lfloor n / 2 \rfloor - 1\right)^2 &\le \log_6 \left((2n)^{k + 1}\right) = \frac{(k + 1) \log 2n}{\log 6}, \\
\frac{\left(\lfloor n / 2 \rfloor - 1\right)^2 \log 6}{\log 2n} - 1 &\le k.
\end{align*}
Hence, there must exist some configurations
which are $\Omega(n^2 / \log n)$ moves away from solved.
\qed
\end{proof}
}


\section{Diameter of $n \times n \times n$ Rubik's Cube}
\label{nxnxn}

Because the only visible cubies
on the $n \times n \times n$ Rubik's Cube
are on the surface,
we use an alternative coordinate system.
Each cubie has a face coordinate
$(x, y) \in \group{n} \times \group{n}$.
Consider the set of reachable locations for
a cubie with coordinates $(x, y) \in \group{n} \times \group{n}$
on the front face.
A face rotation of the front face
will let it reach the coordinates
$(n - y - 1, x)$, $(n - x - 1, n - y - 1)$, and $(y, n - x - 1)$
on the front face.
Row or column moves will allow the cubie to move to another face,
where it still has to have one of those four coordinates.
Hence, it can reach 24 locations in total.
For this problem, we define the \CCS{} $(x, y)$
to be those 24 positions that are reachable
by the cubie $(x, y)$.

\later{
We define \concept{edge cubies}
to be cubies with more than one face visible.
If a cluster has an edge cubie,
then all of its cubies are edge cubies.
We call such clusters \emph{edge clusters}.
We define \concept{corner cubies}
to be cubies with more than two faces visible.
All corner cubies are in a single cluster
known as the \concept{corner cluster}.
If $n$ is odd,
we must also define several other types of cubies.
We first define \concept{cross cubies}
to be cubies with face coordinates of the form
$(x, (n - 1) / 2)$ or $((n - 1) / 2, y)$.
If a cluster contains a cross cubie,
then all of its cubies are cross cubies,
and the cluster is called a \emph{cross cluster}.
We define \concept{center cubies}
to be the six cubies with face coordinates
$((n - 1) / 2, (n - 1) / 2)$.
They form a special cluster
which we will call the \concept{center cluster}.
Our goal in solving the Rubik's Cube
will be to make each side match the color
of its center cluster.
Hence, there is no need to solve
the center cluster.

Given a particular \CCSconfig{},
this configuration can be converted to
the solved color configuration
by performing a sequence of pairwise cubie swaps.
If an order is placed on the cubies in the cluster,
as in Figure~\ref{fig:perm_order},
each pairwise cubie swap is a permutation of this order,
and the set of all cubies swaps
generates the permutation group on 24 elements, $S_{24}$.
By permutation group theory,
if an even (odd) number of swaps
can be applied to a color configuration
to transform it to the solved color configuration,
then \emph{any} sequence of swaps
transforming a configuration
to the solved configuration has an even (odd) number of swaps. 
We call this the \concept{parity} of a color configuration.
}

Just as it was for the $n \times n \times 1$ cube,
our goal is to prove that
for each cluster configuration,
there is a sequence of $O(1)$ moves
which can be used to solve the cluster,
while not affecting any other clusters.
For the $n \times n \times 1$ cube,
we wrote these solution sequences
using the symbols $H_1, H_2, V_1, V_2$
to represent a general class of moves,
each of which could be mapped to a specific move
once the cubie cluster coordinates were known.
Here we introduce more formal notation.

Because of the coordinate system we are using,
we distinguish two types of legal moves.
\concept{Face moves} involve taking a single face
and rotating it $90^\circ$ in either direction.
\concept{Row or column moves}
involve taking a slice of the cube
(not one of its faces)
and rotating the cubies in that slice
by $90^\circ$ in either direction.
Face moves come in twelve types,
two for each face.
For our purposes,
we will add a thirteenth type which applies the identity function.
If $a$ is the type of face move,
we write $F_{a}$ to denote the move itself.
Given a particular index $i \in \{1, 2, \ldots, \lfloor n/2 \rfloor - 1\}$,
there are twelve types of row and column moves that can be performed
--- three different axes for the slice,
two different indices ($i$ and $n - i - 1$) to pick from,
and two directions of rotation.
Again, we add a thirteenth type
which applies the identity function.
If $a$ is the type of row or column move,
and $i$ is the index,
then we write $RC_{a, i}$ to denote the move itself.

A \concept{cluster move sequence}
consists of three type sequences:
face types $a_1, \ldots, a_\ell$,
row and column types $b_1, \ldots, b_\ell$,
and row and column types $c_1, \ldots, c_\ell$.
For a cubie cluster $(x, y)$,
the sequence of actual moves
produced by the cluster move sequence
is $F_{a_1}, RC_{b_1, x}, RC_{c_1, y}, \ldots,
F_{a_\ell}, RC_{b_\ell, x}, RC_{c_\ell, y}$.
A \concept{cluster move solution} for a cluster configuration $d$
is a cluster move sequence with the following properties:
\begin{enumerate}

\item
\label{property:pseudo-solves}
For any
$(x, y) \in \{1, 2, \ldots, \lfloor n / 2 \rfloor - 1\} \times
\{1, 2, \ldots, \lfloor n / 2 \rfloor - 1\}$, if
cluster $(x, y)$ is in configuration $d$, then it can be solved using
the sequence of moves $F_{a_1}, RC_{b_1, x}, RC_{c_1, y}, \ldots,
F_{a_\ell}, RC_{b_\ell, x}, RC_{c_\ell, y}$.

\item
\label{property:pseudo-mirror}
The move sequence $F_{a_1}, RC_{b_1, x}, RC_{c_1, y}, \ldots, F_{a_\ell}, RC_{b_\ell, x}, RC_{c_\ell, y}$
does not affect cubie cluster $(y, x)$.

\item
\label{property:pseudo-other}
All three of the following sequences of moves do not affect the
configuration of any cubie clusters:
\begin{align*}
&F_{a_1}, RC_{b_1, x}, F_{a_2}, RC_{b_1, x}, \ldots, F_{a_\ell}, RC_{b_\ell, x}; \\
&F_{a_1}, RC_{c_1, y}, F_{a_2}, RC_{c_1, y}, \ldots, F_{a_\ell}, RC_{c_\ell, y}; \\
&F_{a_1}, F_{a_2}, \ldots, F_{a_\ell}.
\end{align*}
\end{enumerate}
Our goal is to construct a cluster move solution
for each possible cubie cluster configuration,
and then use the cluster move solution
to solve multiple cubie clusters in parallel.

In the speed cubing community,
there is a well-known technique for solving $n \times n \times n$ Rubik's Cubes
in $O(n^2)$ moves,
involving a family of constant-length cluster move sequences.
These cluster move sequences
can be combined to construct
constant-length cluster move solutions
for all possible cluster configurations,
which is precisely what we wanted.

\later{
Before dealing with the general solution, we address fixing the parity of the cubie clusters.
This allows us to assume that the parity of all clusters is even for the remainder of the paper.

\begin{figure}
\centering
\includegraphics[scale=0.8]{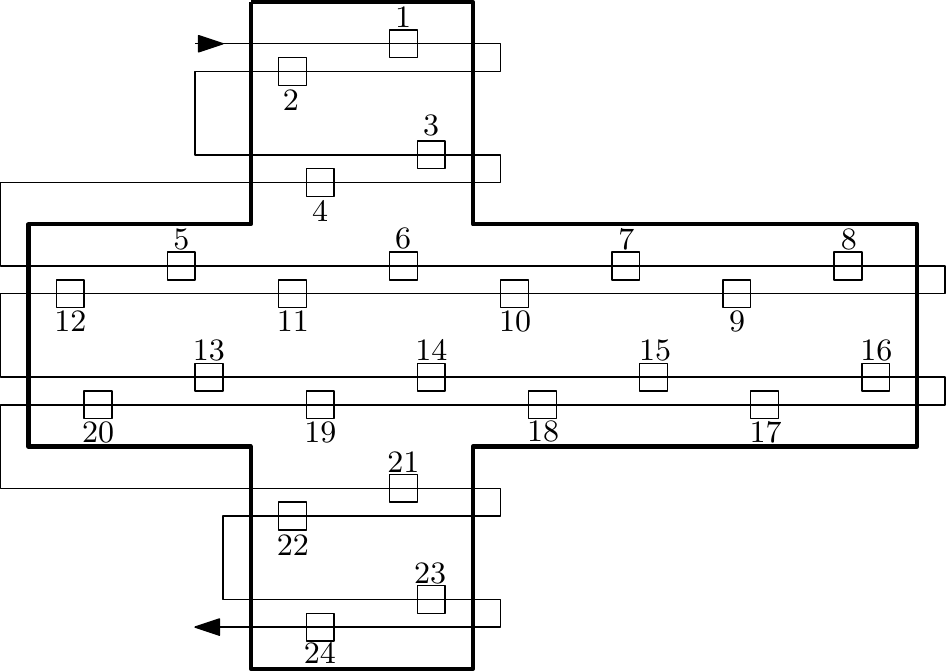}
\caption{An ordering of a 24-\CCS{}.}
\label{fig:perm_order}
\end{figure}

\begin{lemma}
\label{lemma:even-parity-linear}
Given a solvable $n \times n \times n$ Rubik's Cube configuration, the parity of all \CCSs{} can be made even in $O(n)$ moves.
\end{lemma}

\begin{proof}
By definition,
the center cluster is already solved,
and therefore we may assume that 
its parity is already even.
In addition, any cluster
containing at least two indistinguishable cubies
can be considered to have even parity or odd parity
depending on the chosen label for the indistinguishable cubies.
Therefore, we may assume that all such clusters
have even parity.
This means that all non-edge clusters,
including the non-edge cross clusters,
can be assumed to have the correct parity
no matter how many moves are performed.
So we need only fix
the parity of the edge clusters.

We begin by fixing the parity of the corner cluster
and the edge cross cluster (if it exists).
Because the cube is solvable,
we know that the corner cluster
and the edge cross cluster can be solved.
Because the corner cluster
has $O(1)$ reachable states
and the edge cross cluster has $O(1)$ reachable states,
we know that we can solve both in $O(1)$ moves.
Once those two clusters are solved,
we know that their parities must be correct.
Therefore, there is a sequence of $O(1)$ moves
which can be used to fix the parity of those clusters.

Consider the effect of a face move
on the parity of a non-cross edge cluster.
For a particular edge cluster,
a face move affects the location of eight cubies,
due to the fact that a face move
also acts like a row or column move
for edge cubie groups.
The color of each cubie is rotated $90^{\circ}$
in the direction of the face's rotation.
This means that the permutation applied
consists of two permutation cycles
each containing four elements.
Therefore,
if the elements whose colors are changed
are $1, 2, 3, \ldots, 8$,
then we can write the applied permutation as
$(1~3)(1~5)(1~7)(2~4)(2~6)(2~8)$,
or six swaps.
Hence face moves cannot be used to fix 
the parity of the edge clusters.

Now consider the effect of a row or column move
on the parity of a non-cross edge cluster.
A row or column move
affects the colors of four cubies,
one for each corner of the rotated slice.
The color of each cubie
is transferred to the adjacent cubie
in the direction of the move rotation.
So if the elements whose colors are changed are $1, 2, 3, 4$,
then the applied permutation is $(1~2~3~4) = (1~2)(1~3)(1~4)$.
Because the permutation can be written as an odd number of swaps,
the parity of the cluster has changed.
Note, however, that there is exactly one edge cluster
whose parity is affected by this movement.
Therefore,
we can correct the parity of each odd edge cluster
by performing a single row or column move
that affects the cluster in question.
The total number of moves required
is therefore proportional to the number of edge clusters,
or $O(n)$.
\qed
\end{proof}

\begin{lemma}
\label{lemma:invariant-atom}
For each permutation in Table~\ref{table:invariant-moves}, there exists a cluster move sequence of length $O(1)$ which can be used to apply the given permutation to the cluster while applying the identity permutation to every other cluster.
\end{lemma}

\begin{proof}
First, we introduce some more notation for Rubik's Cube moves.
It is more specific than
the formalism of the cluster move sequence introduced above,
making it easier to express a particular set of moves,
but is more difficult to analyze
when we parallelize this move sequence.

Consider facing the cube from in front (the front face is the face in the $xz$-plane with the more negative $y$-value).
From this view, there are horizontal moves that rotate a slice of the cube parallel to the $xy$-plane.
Rotating the $i$th slice from the top in the clockwise direction $90^{\circ}$ as viewed from above the cube is denoted by $H_i^{CW}$.
Rotating this same slice in the opposite direction is denoted by $H_i^{CCW}$.
Similarly, rotating a slice parallel to the $yz$-plane is a vertical move, and rotating the $j$th slice from the left side of the cube in the clockwise direction as viewed from left of the cube is denoted $V_j^{CW}$, while the counter-clockwise version is denote $V_j^{CCW}$.
Finally, we define a third type of move which rotates a slice parallel to the front face.
Counting inward from the front face, we denote rotating the $k$th slab $90^{\circ}$ clockwise as $D_k^{CW}$, while rotating it $90^{\circ}$ in the opposite direction is $D_k^{CCW}$.
See Figure~\ref{fig:nxnxn_moves}.

\begin{figure}
\centering
\includegraphics[scale=0.7]{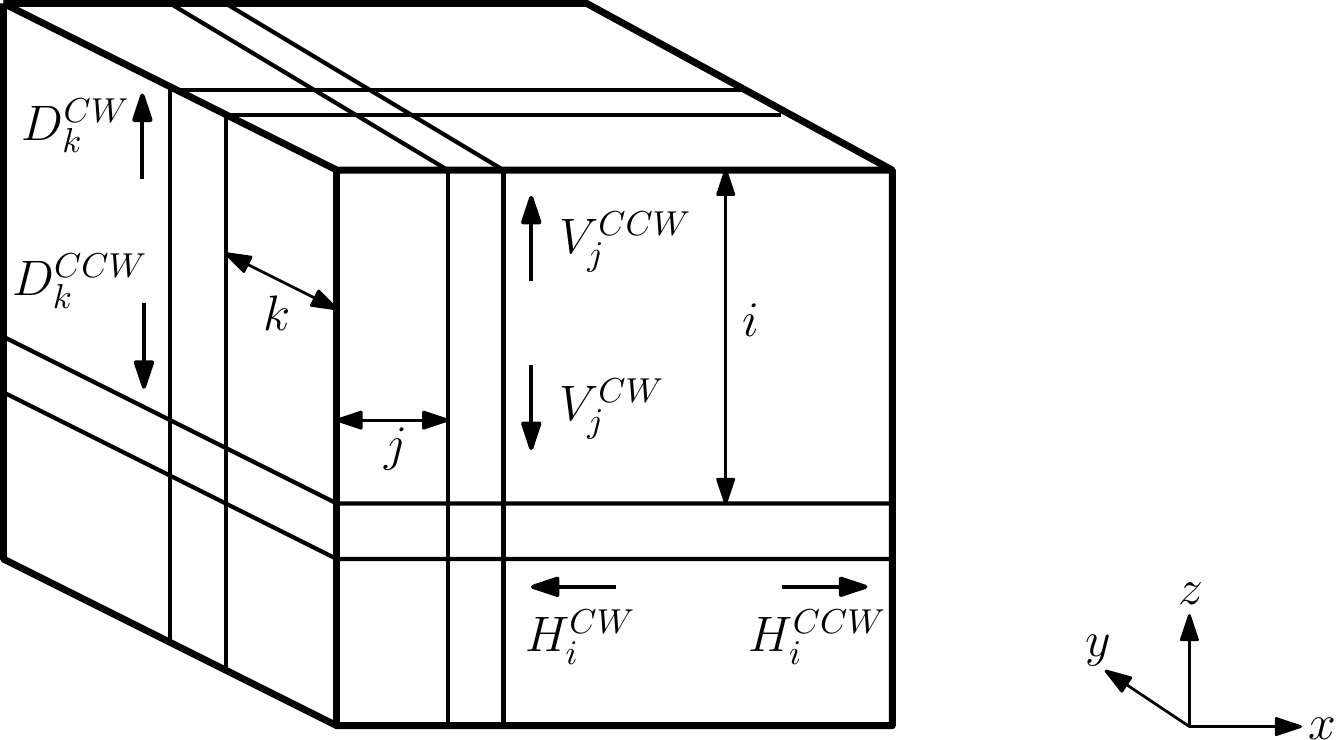}
\caption{The definitions of the various moves defined on an $n \times n \times n$ Rubik's Cube.}
\label{fig:nxnxn_moves}
\end{figure}

We claim that the move sequence $S = V_n^{CCW} \circ D_m^{CW} \circ V_n^{CW} \circ D_m^{CCW} \circ H_0^{CW} \circ D_m^{CW} \circ V_n^{CCW} \circ D_m^{CCW} \circ V_n^{CW} \circ H_0^{CCW}$ swaps the colors of three cubies in a single cluster, while leaving the color configurations of all other clusters the same.
This move sequence is attributed to Ingo Sch\"{u}tze \cite{schutze}, but is fairly well-known within the speed cubing community, so its origins are unclear.
The effect of applying $S$ is shown by case analysis of individual cubies.
When applying the sequence $S$, only cubies lying on the union of the slices rotated by the $V_{n}$, $D_n$, and $H_0$ moves are affected.

Blocks lying on the bottom face are unaffected, as they never reach the top face and thus only have the subsequence $V_n^{CCW} \circ D_m^{CW} \circ V_n^{CW} \circ D_m^{CCW} \circ D_m^{CW} \circ V_n^{CCW} \circ D_m^{CCW} \circ V_n^{CW}$ applied, which does not affect the final location of a cube. 
Blocks starting on the back and right faces never move to the upper face, as each move that could place these on the upper face (moves $V_n^{CW}$ and $D_m^{CCW}$) is preceded by a move rotating these cubies onto the bottom face (moves $V_n^{CCW}$ and $D_m^{CW}$).

Now consider the cubies on left face in the slice rotated by $D_m^{CW}$ and $D_m^{CCW}$.
Exactly one of these cubies is in the same cluster as the cubie on the upper face that lies in the slices rotated by \emph{both} $V_n^{CW}$ and $D_m^{CW}$.
All other such cubies cannot be affected by $D_m^{CW}$ and $D_m^{CCW}$ moves, so applying $S$ has the same effect as applying the sequence $D_m^{CW} \circ D_m^{CCW} \circ H_0^{CW} \circ D_m^{CW} \circ D_m^{CCW} \circ H_0^{CCW}$.
Canceling the $D_m^{CW}$ and $D_m^{CCW}$ moves yields $H_0^{CW}, H_0^{CCW}$ and thus these cubies are unaffected.
Now consider the single cubie on the left face in the same cluster as the cubie on the upper face that lies in the slices rotated by \emph{both} $V_n^{CW}$ and $D_m^{CW}$.
Tracing the locations visited by this cubie when $S$ is applied to it shows that the cubie travels to the upper face (via $D_m^{CW}$), the front face (via $L_n^{CW}$), the upper face (via $V_n^{CCW}$), and then the left face (via $D_m^{CCW}$).
So the cubie's location is unaffected by $S$.

Next consider the cubies on the front face.
Exactly one of these cubies is in the same cluster as the cubie on the upper face that lies in the slices rotated by both $V_n^{CW}$ and $D_m^{CW}$.
All other such cubies cannot be affected by $V_n^{CW}$ and $V_n^{CCW}$ moves, so applying $S$ has the same effect as applying the sequence $D_m^{CW} \circ D_m^{CCW} \circ H_0^{CW} \circ D_m^{CW} \circ D_m^{CCW} \circ H_0^{CCW}$. 
For the cubie in the same cluster as the cubie on the upper face that lies on both the $V_n^{CW}$ and $D_m^{CW}$ slices, applying $S$ to it results in moving it, in sequence, to the left side of the upper face ($V_n^{CCW}$), right face ($D_m^{CW}$), upper face ($D_m^{CCW}$), the back side of the upper face ($H_0^{CW}$), and the left side of the upper face ($H_0^{CCW}$).
So applying $S$ to this cubie moves it to the location of the cubie in its cluster in the left side of the upper face. 

Finally, consider the cubies on the upper face.
Divide the cubies into three sets based upon the cubies in their clusters.
Each cubie on the upper face either is in the cluster containing the cubie lying in slices rotated by both $V_n^{CW}$ and $D_m^{CW}$, or is in a cluster containing a cubie lying in exactly one of the slices rotated by $V_n^{CW}$ and $D_m^{CW}$, or is in a cluster that does not contain any elements in either of the slices rotated by $V_n^{CW}$ and $D_m^{CW}$.
If a cubie lies in a cluster that does not contain any elements in either of the slices rotated by $V_n^{CW}$ and $D_m^{CW}$, then it cannot be affected by $V_n^{CW}$, $V_n^{CCW}$, $D_m^{CW}$, or $D_m^{CCW}$ moves.
So applying $S$ to it is equivalent to applying $H_0^{CW}, H_0^{CCW}$ to it, and thus does not affect its position.
If a cubie lies in the $V_n^{CW}$ slice but is not in the cluster containing the cubie in both $V_n^{CW}$ and $D_m^{CW}$ slices, then applying $S$ to it is equivalent to applying $V_n^{CCW} \circ V_n^{CW} \circ H_0^{CW} \circ V_n^{CCW} \circ V_n^{CW} \circ H_0^{CCW}$ (the identity) as it can never lie in the $D_m^{CW}$ slice. 
Similarly, if a cubie lies in the $D_m^{CW}$ slice but is not in the cluster containing the cubie in both $V_n^{CW}$ and $D_m^{CW}$ slices, then it can never lie in the $V_n^{CW}$ slice, and so applying $S$ to it is equivalent to applying $D_m^{CW} \circ D_m^{CCW} \circ H_0^{CW} \circ D_m^{CW} \circ D_m^{CCW} \circ H_0^{CCW}$, the identity.

Now consider the four cubies on the upper face in the same cluster as the cubie lying in both $V_n^{CW}$ and $D_m^{CW}$ slices.
The cubie lying on the left side of the upper face \emph{is} the cubie lying in both the $V_n^{CW}$ and $D_m^{CW}$ slices, and applying $S$ to it results in moving it, in sequence, to the upper side of the back face ($V_n^{CCW}$), the left side of the upper face ($V_n^{CW}$), the bottom side of the left face ($D_m^{CCW}$), the left side of the upper face ($D_m^{CW}$), the upper side of the back face ($V_n^{CCW}$), the left side of the upper face ($V_n^{CW}$), and the front side of the upper face ($H_0^{CCW}$).
So the result is moving the cubie from the left side to the front side of the upper face.
The cubie lying on the front side of the upper face initially lies in neither the $V_n^{CW}$ nor the $D_m^{CW}$ slice.  
So the moves in $S$ before $H_0^{CW}$ do no affect it.
Applying the subsequence of moves starting at $H_0^{CW}$ move it, in sequence, to the left side of the upper face ($H_0^{CW}$), the upper side of the right face ($D_m^{CW}$), the left side of the upper face ($D_m^{CCW}$), and the left side of the front face ($V_n^{CW}$).
So the result is moving the cubie from the front side of the upper face to the left side of the front face.
The cubie lying on the back side of the upper face moves visits the back face ($V_n^{CCW}$, $V_n^{CW}$) and the right side of the upper face ($H_0^{CW}$, $H_0^{CCW}$), but is not affected by $S$, and the cubie lying on the right side of the upper face only visits the front side of the upper face ($H_0^{CW}$, $H_0^{CCW}$) and thus is not affected by $S$.

In summary, applying $S$ to the cube results in changing the locations of exactly three cubies of a single cluster, those lying on the left and front sides of the upper face, and the cubie lying on the left side of the front face.
All other cubies of the cube are left unchanged.
The three affected cubies each move into the location of another, with the cubie on the left side of the upper face moving to the location of the cubie on the front side of the upper face, the cubie on the front side of the upper face moving to the location of the cubie on the left side of the front face, and the cubie on the left side of the front face moving to the location of the cubie on the left side of the upper face.
As seen in Figure~\ref{fig:triple_rotate}.
the result of applying $S$ to a cluster is to ``rotate'' the locations of three cubies, and in effect rotate the colors of the cubies at these three locations.
Note that rotating three elements is equivalent to performing a pair of swaps, just as the permutation $(11~19~22) = (11~19)(11~22)$

The choices for which faces are front, left and upper are arbitrary, and there are 24 choices for such a set (six choices for the front face, and four choices for the upper face for each choice of front face).
For a specific cluster, each choice of front, left and upper faces implies a permutation resulting from applying $S$.
Using the ordering of the cubies in a cluster defined in Figure~\ref{fig:perm_order}, a resulting set of 24 permutations is generated (as seen in Table~\ref{table:invariant-moves}).
\qed
\end{proof}

\begin{figure}[t]
\centering
\includegraphics[scale=0.8, trim=0.77in 0in 0in 0in, clip]{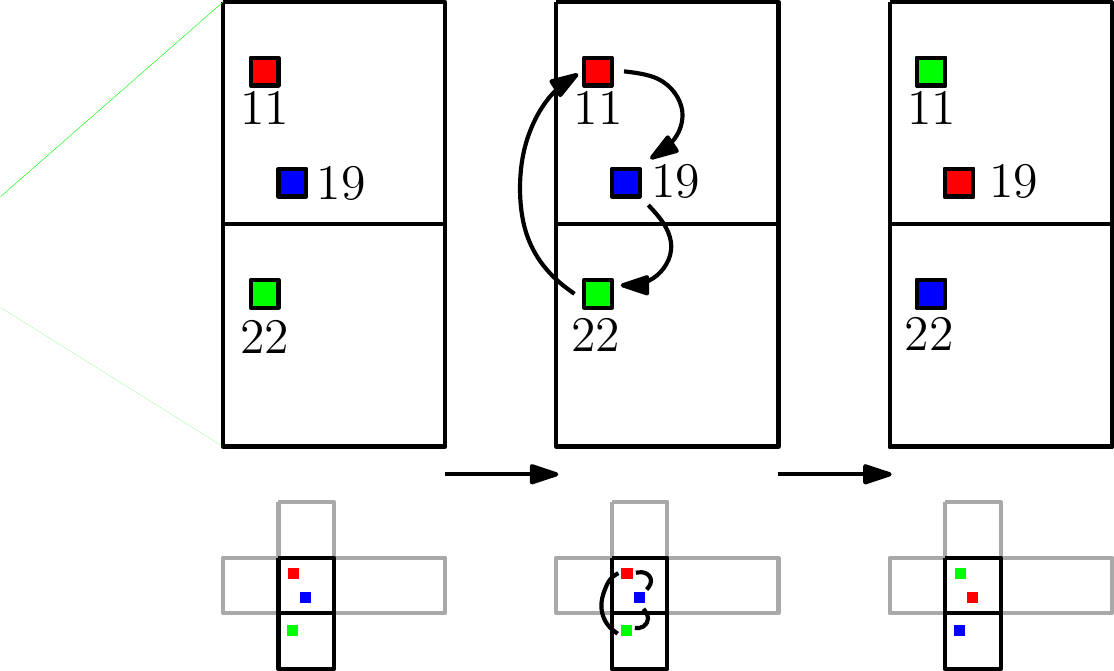}
\caption{The resulting movement of the three cubies of a cluster resulting from applying the move sequence $S = V_n^{CCW} \circ D_m^{CW} \circ V_n^{CW} \circ D_m^{CCW} \circ H_0^{CW} \circ D_m^{CW} \circ V_n^{CCW} \circ D_m^{CCW} \circ V_n^{CW} \circ H_0^{CCW}$.}
\label{fig:triple_rotate}
\end{figure}

\begin{table}[t]
\centering
\begin{tabular}{| c c c c c c |}
\hline
$(1~2~12)$ & $(4~3~10)$ & $(2~4~11)$ & $(3~1~9)$ & $(5~12~8)$ & $(20~13~19)$ \\ $(12~20~24)$ & $(13~5~4)$ & $(6~11~5)$ & $(19~14~18)$ & $(11~19~22)$ & $(14~6~3)$ \\
$(7~10~6)$ & $(18~15~17)$ & $(10~18~21)$ & $(15~7~1)$ & $(8~9~7)$ & $(17~16~20)$ \\
$(9~17~23)$ & $(16~8~2)$ & $(21~22~13)$ & $(24~23~15)$ & $(22~24~16)$ & $(23~21~14)$ \\
\hline
\end{tabular}
\caption{A set of 24 permutations that can be applied to a cluster while leaving all other clusters unchanged.}
\label{table:invariant-moves}
\end{table}

\begin{lemma}
\label{lemma:nxnxn-group-reachable}
Any cluster configuration with even parity
can be solved using a cluster move solution
of length $O(1)$.
\end{lemma}

\begin{proof}
By Lemma~\ref{lemma:invariant-atom}, there exist a set of permutations that can be applied to any single cluster while applying the identity permutation to every other cluster.
It can be shown using the GAP software package \cite{GAP} that this set of permutations generates $A_{24}$, the set of even permutations on 24 elements.
Thus any even permutation can be written as a composition of these permutations and has an inverse that can also be written as the composition of these permutations.
Because each cluster has finite size, the inverse composition must have $O(1)$ length.
So there exists a $O(1)$-length sequence of moves that can be applied to the cube that results in one cluster having the solved color configuration, and all other \CCSs{} having unchanged color configurations.
\qed
\end{proof}
}

\subsection{$n \times n \times n$ Upper Bound}

As in the $n \times n \times 1$ case,
our goal here is to find a way to solve several different clusters in parallel,
so that the number of moves for the solution
is reduced from $O(n^2)$ to $O(n^2 / \log n)$.
\ifabstract
\else
We use the same parallelization technique
as we did for the $n \times n \times 1$ Rubik's Cube,
although some modification is necessary
because of differences in the types of moves allowed.
\fi

\later{
\begin{lemma}
\label{lemma:nxnxn-bulk-same}
Suppose we are given
an $n \times n \times n$ Rubik's Cube configuration
and sets $X, Y \subseteq \group{\floor{n / 2}}$
such that $X \cap Y = \emptyset$.
If all \CCSs{} in $X \times Y$
have the same cluster configuration,
then they can all be solved in
a sequence of $O(|X| + |Y|)$ moves
that only affects \CCSs{} in $(X \times Y) \cup (X \times X) \cup (Y \times Y)$.
\end{lemma}
}

\later{
\begin{proof}
Let $d$ be the cluster configuration
of all of the clusters in $X \times Y$.
By Lemma \ref{lemma:nxnxn-group-reachable},
we know that there is
a constant-length cluster move solution for $d$.
Let $a_1, \ldots, a_m$,
$b_1, \ldots, b_m$,
and $c_1, \ldots, c_m$
be the type sequences of that cluster move solution.
Let $x_1, \ldots, x_\ell$ be the elements of $X$,
and let $y_1, \ldots, y_k$ be the elements of $Y$.
To build a sequence of moves
to solve all the clusters in $X \times Y$,
we begin by defining:
\begin{align*}
\textsc{Bulk}_i = F_{a_i},
RC_{b_i, x_1}, RC_{b_i, x_2}, \ldots, RC_{b_i, x_\ell},
RC_{c_i, y_1}, RC_{c_i, y_2}, \ldots, RC_{c_i, y_k}.
\end{align*}
Note that this sequence
consists of $|X| + |Y| + 1$ moves.
We then construct the full sequence of moves
to be the following:
\begin{align*}
\textsc{Bulk}_1, \textsc{Bulk}_2, \ldots, \textsc{Bulk}_m.
\end{align*}
Because the original sequence of moves had length $O(1)$,
we know that $\ell = O(1)$,
and so the total length of this sequence
will be $O(|X| + |Y| + 1)$.

Consider the effect of this constructed move sequence
on a cubie cluster
$(x, y) \in \group{\floor{n / 2}} \times \group{\floor{n / 2}}$.
First, consider the effect
on $(x, y) \in X \times Y$.
The subsequence of moves which affect this \CCS{} will be
$F_{a_1}, RC_{b_1, x}, RC_{c_1, y}, \ldots,
F_{a_\ell}, RC_{b_\ell, x}, RC_{c_\ell, y}$.
This is precisely the set of moves
generated by the cluster move solution
for solving the cluster $(x, y)$
so this cluster will be solved.
The subsequence of moves
affecting the cluster
$(y, x) \in Y \times X$.
will be the same as the subsequence
for the cluster $(x, y)$.
By Property~\ref{property:pseudo-mirror}
of cluster move solutions,
this cluster will not be affected
by the sequence of moves.

We need not consider
the effect on clusters $X \times X$
or $Y \times Y$,
because our lemma places no restrictions
on what happens to those clusters.
So all of the remaining clusters
we must consider
have at most one coordinate
in $X \cup Y$.
Suppose we have some $x \in X$
and some $z \notin X \cup Y$.
Then the sequence of moves
affecting the clusters $(x, z)$ and $(z, x)$
will be 
$F_{a_1}, RC_{b_1, x}, F_{a_2}, RC_{b_2, x} \ldots,
F_{a_\ell}, RC_{b_\ell, x}$.
By Property~\ref{property:pseudo-other}
of cluster move solutions,
this sequence of moves
does not affect any clusters,
and so $(x, z)$ and $(z, x)$ will both remain unaffected.
Similarly, suppose we have some $y \in Y$
and some $z \notin X \cup Y$.
Then the sequence of moves
affecting the clusters $(y, z)$ and $(z, y)$ is
$F_{a_1}, RC_{c_1, y}, F_{a_2}, RC_{c_2, y} \ldots,
F_{a_\ell}, RC_{c_\ell, y}$.
According to Property~\ref{property:pseudo-other},
this move sequence does not affect
the configuration of clusters $(y, z)$ or $(z, y)$.
Finally, consider the effect
on some cluster
$(w, z) \in \overline{X \cup Y} \times \overline{X \cup Y}$.
Then the sequence of moves
affecting $(w, z)$ is
$F_{a_1}, F_{a_2}, \ldots, F_{a_\ell}$.
Once again, by Property~\ref{property:pseudo-other}
of cluster move solutions,
this move sequence will not affect
the configuration of cluster $(w, z)$.
\qed
\end{proof}
}

Now say that we are given a cluster configuration $d$
and a set of columns $X$ and rows $Y$
such that $X \cap Y = \emptyset$.
Using the same row-grouping technique
that we used for the $n \times n \times 1$ case,
\ifabstract
\else
we show the following lemma.
\fi

\later{
\begin{lemma}
\label{lemma:nxnxn-bulk-subsets}
Suppose we are given
an $n \times n \times n$ Rubik's Cube configuration,
a \CCSconfig{} $c$,
and sets $X, Y \subseteq \group{\floor{n / 2}}$,
such that $\ell = |X|$
and $X \cap Y = \emptyset$.
Then there exists a sequence of moves
of length $O(\ell 2^\ell + |Y|)$
such that:
\begin{itemize}

\item
All \CCSs{} $(x, y) \in X \times Y$ in configuration $c$
will be solved.

\item
All \CCSs{} $(x, y) \in (X \times X) \cup (Y \times Y)$
may or may not be affected.

\item
All other \CCSs{} will not be affected.

\end{itemize}
\end{lemma}
}

\later{
\begin{proof}
For each row $y \in Y$,
let $S_y = \{ x \in X \mid $ \CCS{} $(x, y)$
is in configuration $c\}$.
For each set $S \subseteq X$,
let $Y_S = \{ y \in Y \mid S_y = S \}$.
Because $S \subseteq X$,
there are at most $2^\ell$ different values for $S$.
For each $S$, 
we will use the results of Lemma \ref{lemma:nxnxn-bulk-same}
to construct a sequence of moves
to solve each \CCS{} $(x, y) \in S \times Y_S$.
This move sequence will have length
$O(|S| + |Y_S|) = O(\ell + |Y_S|)]$.
When we sum up this cost
over all sets $S \subseteq X$,
we get the following number of moves:
\begin{align*}
O\left(\ell \cdot 2^\ell + \sum_{S} |Y_S|\right) = O\left(\ell \cdot 2^\ell + |Y|\right).
\end{align*}
\qed
\end{proof}
}

\iffull
Just as we did for the $n \times n \times 1$ Rubik's Cube,
we avoid exponential blow-up by dividing the set of columns $X$
into smaller groups, solving each such group individually.
More formally:
\fi

\later{
\begin{lemma}
\label{lemma:nxnxn-bulk-et-al}
Suppose we are given
an $n \times n \times n$ Rubik's Cube configuration,
a \CCSconfig{} $c$,
and sets $X, Y \subseteq \group{\floor{n / 2}}$,
such that $X \cap Y = \emptyset$.
Then there exists a sequence of moves
of length $O(|X| \cdot |Y| / \log |Y|)$
such that:
\begin{itemize}

\item
All \CCSs{} $(x, y) \in X \times Y$ in configuration $c$
will be solved.

\item
All \CCSs{} $(x, y) \in (X \times X) \cup (Y \times Y)$
may or may not be affected.

\item
All other \CCSs{} will not be affected.

\end{itemize}
\end{lemma}
}

\later{
\begin{proof}
Let $\ell = \frac{1}{2} \log_2 |Y|$,
so that $2^\ell = \sqrt{|Y|}$.
Let $k = \lceil |X| / \ell \rceil$.
Partition the set $X$ into
a series of sets $X_1, \ldots, X_{k}$
each of which has size $\le \ell$.
For each $X_i$,
we solve the \CCSs{} in $X_i \times Y$
using the sequence of moves that is guaranteed to exist
by Lemma \ref{lemma:nxnxn-bulk-subsets}.
The number of moves required to solve
a single $X_i$ is
\begin{align*}
O\left(\ell 2^\ell + |Y| \right) = O\left(\left(\frac{1}{2} \log_2 |Y|\right) \sqrt{|Y|} + |Y| \right) = O(|Y|).
\end{align*}
Therefore, if we wish to perform this for $k$ sets,
the total number of moves becomes
\begin{align*}
O\left(k \cdot |Y|\right) = O\left( \frac{|X|}{\frac{1}{2} \log_2 |Y|} \cdot |Y|\right) = O\left( \frac{|X| \cdot |Y|}{ \log |Y|} \right).
\end{align*}
\qed
\end{proof}
}

To finish constructing the move sequence for the entire Rubik's Cube,
we need to account for two differences
between this case and the $n \times n \times 1$ case:
the requirement that $X \cap Y = \emptyset$
and the potential to affect clusters in
$(X \times X) \cup (Y \times Y)$.

\both{
\begin{theorem}
\label{theorem:nxnxn-upper-bound}
Given an $n \times n \times n$ Rubik's Cube configuration,
all \CCSs{} can be solved in $O(n^2 / \log n)$ moves.
\end{theorem}
}

\later{
\begin{proof}
In order to solve the Rubik's Cube,
we must first fix the parity.
Using the techniques of Lemma~\ref{lemma:even-parity-linear},
we can perform this step in $O(n)$ moves.
Then we solve each edge cluster individually.
Each edge cluster requires $O(1)$ moves to solve,
and there are $O(n)$ edge clusters,
so this preliminary step
takes time $O(n)$.

Once the edges have been solved,
we want to solve the non-edge clusters.
Let $k = \sqrt{n / 2}$.
Partition $\group{\floor{n / 2}}$ into
a series of sets $G_1, \ldots, G_k$
each of which has size $\le k$.
For each pair $i, j$ such that $i \ne j$
and each \CCSconfig{} $c$,
we use the sequence of moves
guaranteed to exist by Lemma \ref{lemma:nxnxn-bulk-et-al}
to solve all $(x, y) \in G_i \times G_j$
with the configuration $c$.
This ensures that all \CCSs{} $(x, y) \in G_i \times G_j$
will be solved.
For each $i$,
we must also solve all \CCSs{}
$(x, y) \in G_i \times G_i$.
There are $k^2(k - 1) / 2$ such \CCSs{},
so we can afford to solve each such \CCS{} individually.

What is the total number of moves required?
For a single pair $i \ne j$
and a single configuration $c$,
the number of moves required
will be $O(k^2 / \log k) = O(n / \log n)$.
There are a constant number of possible configurations,
so solving a single pair $i, j$ for all configurations
will also require $O(n / \log n)$ moves.
There are $k^2 - k$ such pairs,
for a total of $O(n^2 / \log n)$.
If we then add in the extra $O(k^3)$
from the diagonals,
then the total number of moves
will be $O(n^2 / \log n + n^{3/2}) = O(n^2 / \log n)$.
\qed
\end{proof}
}

\subsection{$n \times n \times n$ Lower Bound}

\ifabstract
\else
We derive a matching lower bound
using a technique identical to the one used
for the $n \times n \times 1$ lower bound:
\fi

\both{
\begin{theorem}
\label{theorem:nxnxn-lower-bound}
Some configurations
of an $n \times n \times n$ Rubik's Cube
are $\Omega(n^2 / \log n)$ moves away from being solved.
\end{theorem}
}

\later{
\begin{proof}
Lemma \ref{lemma:nxnxn-group-reachable} shows that
for every possible configuration of a \CCS{},
there exists a sequence of moves to solve the \CCS{}
while leaving the rest of the cubies in the same location.
Hence,
the inverse of such a sequence
will transform a solved \CCS{} to an arbitrary configuration
without affecting any other cubies.
Not counting the edges,
there are $\left(\lfloor n / 2 \rfloor - 1\right)^2$ \CCSs{},
each of which can be independently placed
into one of $(24!) / (4!)^6$ different color configurations.
This means that there are at least
\begin{align*}
\left(\frac{24!}{(4!)^6}\right)^{\left(\lfloor n / 2 \rfloor - 1\right)^2}
\end{align*}
reachable configurations.

There are $6n$ possible moves, so
the total number of states reachable
using at most $k$ moves
is at most
\begin{align*}
\frac{(6n)^{k + 1} - 1}{6n - 1} \le (6n)^{k + 1}.
\end{align*}
Therefore,
if $k$ is the number of moves
necessary to reach all states,
it must have the property that
\begin{align*}
\left(\frac{24!}{(4!)^6}\right)^{\left(\lfloor n / 2 \rfloor - 1\right)^2} &\le (6n)^{k + 1}, \\
\left(\lfloor n / 2 \rfloor - 1\right)^2 \cdot  \log \left(\frac{24!}{(4!)^6}\right) &\le \log \left((6n)^{k + 1}\right) = (k + 1) \log (6n), \\
\frac{\left(\lfloor n / 2 \rfloor - 1\right)^2 \log \left(\frac{24!}{(4!)^6}\right)}{\log (6n)} - 1 &\le k.
\end{align*}
Hence, there must exist some configurations
which are $\Omega(n^2 / \log n)$ moves away from solved.
\qed
\end{proof}
}


\section{Optimally Solving a Subset of the $n \times n \times 1$ Rubik's Cube is NP-Hard}
\label{hardness}

In this section, we consider a problem
which generalizes the problem
of computing the optimal sequence of moves
to solve a Rubik's Cube.
Say that we are given a configuration
of an $n \times n \times 1$ Rubik's Cube
and a list of \concept{important} cubies.
We wish to find
the shortest sequence of moves 
that solves the important cubies.
Note that the solution for the important cubies
may cause other cubies
to leave the solved state,
so this problem is only equivalent to solving 
an $n \times n \times 1$ Rubik's Cube
when all cubies are marked important.

In this section, we prove the NP-hardness of computing the length of this shortest sequence.
More precisely, we prove that the following decision problem is NP-hard:
is there a sequence of $k$ moves
that solves the important cubies
of the $n \times n \times 1$ Rubik's Cube?
Our reduction ensures that the cubies within a single \CCS{}
are either all important or all unimportant, and thus it does not matter
whether we aim to solve cubies (which move) or specific cubie positions
(which do not move).
Therefore the problem remains NP-hard if we
aim to solve the puzzle in the sense of unifying the side colors,
when we ignore the colors of all unimportant cubies.

Certain properties of the Rubik's Cube configuration
can affect the set of potential solutions.
For the rest of this section,
we will consider only Rubik's Cubes
where $n$ is odd
and where all edge cubies and cross cubies
are both solved and marked important.
This restriction ensures that for any cluster,
the number of horizontal moves 
and vertical moves affecting it
must both be even.
In addition,
we will only consider Rubik's Cubes
in which all \CCSs{} are in the \CCSconfigs{}
depicted in Figures
\ref{figure:all-red},
\ref{figure:blue-left},
and \ref{figure:blue-up}.
This restriction means
that the puzzle can always be solved
using moves only of types $H_1$ and $V_1$.
This combination of restrictions 
ensures that each unsolved cluster
must be affected by both vertical and horizontal moves.

Suppose that we are given a configuration
and a list of important cubies.
Let $u_r$ be the number of rows
of index $\le \lfloor n / 2 \rfloor$
that contain at least one important unsolved cubie.
Let $u_c$ be the number of columns
of index $\le \lfloor n / 2 \rfloor$
that contain at least one important unsolved cubie.
Then we say that the \concept{ideal number of moves}
for solving the given configuration
is $2(u_r + u_c)$.
In other words, the ideal number of moves
is equal to the smallest possible number of moves
that could solve all the important cubies.
An \concept{ideal solution}
for a subset of the cubies
in a particular $n \times n \times 1$ puzzle
is a solution for that set of cubies
which uses the ideal number of moves.
For the types of configurations that we are considering,
the ideal solution will contain
exactly two of each move,
and the only moves that occur
will be moves of type $H_1$ or $V_1$.

\begin{definition}
Let $I_k(m)$ denote the index in the solution
of the $k$th occurrence of move $m$.
\end{definition}

\ifabstract
\else
For our hardness reduction, we develop two main gadgets.
The first gadget forces an ordering on the second occurrences of row moves,
and is used in the construction of the second gadget.
The second gadget forces a betweenness constraint on the ordering of
the first occurrences of row moves.
\fi

\later{
\begin{lemma}
\label{lemma:gadget-before}
Given two sets of columns
$X_1, X_2 \subseteq \group{\lfloor n / 2 \rfloor}$,
there is a gadget
using three extra rows and two extra columns
ensuring that, for all $x_1 \in X_1$
and for all $x_2 \in X_2$,
$I_2(x_1) < I_2(x_2)$.
As a side effect, this gadget also forces
\begin{align*}
\max_{x_1 \in X_1} I_1(x_1) < \min_{x_1 \in X_1} I_2(x_1)
&& \textrm{ and } &&
\max_{x_2 \in X_2} I_1(x_2) < \min_{x_2 \in X_1} I_2(x_2).
\end{align*}
\end{lemma}
}

\later{
\begin{proof}
Let $\tilde{y}_1, \tilde{y}_2, \tilde{y}_3 \le \lfloor n / 2 \rfloor$ be three rows
not used elsewhere in the construction.
Let $\tilde{x}_1, \tilde{x}_2 \le \lfloor n / 2 \rfloor$ be two columns
not used elsewhere in the construction.
Make \CCSs{} $(\tilde{x}_1, \tilde{y}_2)$ and $(\tilde{x}_2, \tilde{y}_3)$
have the configuration
depicted in Fig.~\ref{figure:blue-left};
make \CCSs{} $(\tilde{x}_1, \tilde{y}_1)$ and $(\tilde{x}_2, \tilde{y}_2)$
have the configuration
depicted in Fig.~\ref{figure:blue-up};
and make \CCS{} $(\tilde{x}_2, \tilde{y}_1)$
remain in the solved configuration.
Mark all of these \CCSs{} as important.

These \CCSconfigs{} enforce the following constraints:
\begin{align*}
I_1(\tilde{x}_1) < I_1(\tilde{y}_2) < I_2(\tilde{x}_1) < I_2(\tilde{y}_2), &&
I_1(\tilde{x}_2) < I_1(\tilde{y}_3) < I_2(\tilde{x}_2) < I_2(\tilde{y}_3), \\
I_1(\tilde{y}_1) < I_1(\tilde{x}_1) < I_2(\tilde{y}_1) < I_2(\tilde{x}_1), &&
I_1(\tilde{y}_2) < I_1(\tilde{x}_2) < I_2(\tilde{y}_2) < I_2(\tilde{x}_2).
\end{align*}
We can use these inequalities
to construct the following chains:
\begin{align*}
I_1(\tilde{y}_1) < I_1(\tilde{x}_1) < I_1(\tilde{y}_2) < I_1(\tilde{x}_2)
&& \textrm{and} &&
I_2(\tilde{y}_1) < I_2(\tilde{x}_1) < I_2(\tilde{y}_2) < I_2(\tilde{x}_2).
\end{align*}
Because $(\tilde{x}_2, \tilde{y}_1)$ must remain solved,
and because of the above constraints,
there is only one possible ordering
for the pairs of moves involving $\tilde{y}_1$ and $\tilde{x}_2$:
$I_1(\tilde{y}_1) < I_2(\tilde{y}_1) < I_1(\tilde{x}_2) < I_2(\tilde{x}_2)$.
If we combine this with the constraint
$I_1(\tilde{x}_2) < I_1(\tilde{y}_3) < I_2(\tilde{x}_2) < I_2(\tilde{y}_3)$,
we know that $I_1(\tilde{y}_1) < I_2(\tilde{y}_1) < I_1(\tilde{y}_3) < I_2(\tilde{y}_3)$.

Now, for each $x_1 \in X_1$
and $x_2 \in X_2$,
make \CCSs{} $(x_1, \tilde{y}_1)$ and $(x_2, \tilde{y}_3)$
have the configuration 
depicted in Fig.~\ref{figure:blue-left}, and mark them important.
No other \CCSs{} involving
$\tilde{y}_1, \tilde{y}_2, \tilde{y}_3$
or $\tilde{x}_1, \tilde{x}_2$
should be marked important.
This constraint ensures that for all $x_1 \in X_1$,
$I_2(x_1)$ must lie between $I_1(\tilde{y}_1)$ and $I_2(\tilde{y}_1)$.
In addition, for all choices of $x_2 \in X_2$,
$I_2(x_2)$ must lie between $I_1(\tilde{y}_3)$ and $I_2(\tilde{y}_3)$.
Therefore, $I_2(x_1) < I_2(x_2)$.

As a side effect,
these constraints ensure that
for all $x_1 \in X_1$,
$I_1(x_1)$ must lie before $I_1(\tilde{y}_1)$,
while $I_2(x_1)$ lies after $I_1(\tilde{y}_1)$.
Therefore, 
\begin{align*}
\max_{x_1 \in X_1} I_1(x_1) < \min_{x_1 \in X_1} I_2(x_1).
\end{align*}
In addition, the constraints ensure that
 for all choices of $x_2 \in X_2$,
$I_1(x_2)$ must lie before $I_1(\tilde{y}_3)$,
while $I_2(x_2)$ lies after $I_1(\tilde{y}_3)$.
This ensures that 
\begin{align*}
\max_{x_2 \in X_2} I_1(x_2) < \min_{x_2 \in X_1} I_2(x_2).
\end{align*}

We have shown that these gadgets can enforce a constraint.
We must also show that these gadgets
do not enforce any constraints
other than the ones expressed in the lemma.
In other words, given any solution
which satisfies the requirements given in the lemma,
we must be able to insert the moves 
for our new rows and columns
in such a way that all important clusters
will be solved.
In order to make sure that clusters
$(\tilde{x}_1, \tilde{y}_2)$,
$(\tilde{x}_2, \tilde{y}_3)$,
$(\tilde{x}_1, \tilde{y}_1)$,
$(\tilde{x}_2, \tilde{y}_2)$,
and $(\tilde{x}_2, \tilde{y}_1)$ are all solved,
it is sufficient to ensure that
the moves $\tilde{x}_1, \tilde{x}_2, \tilde{y}_1, \tilde{y}_2, \tilde{y}_3$
occur in the following order:
\begin{align*}
I_1(\tilde{y}_1), I_1(\tilde{x}_1), I_2(\tilde{y}_1), I_1(\tilde{y}_2), I_2(\tilde{x}_1), I_1(\tilde{x}_2), I_2(\tilde{y}_2), I_1(\tilde{y}_3), I_2(\tilde{x}_2), I_2(\tilde{y}_3).
\end{align*}
So if we can find the correct way to interleave this sequence
with the existing move sequence,
we will have a sequence that solves all clusters.

First, we observe that
the only important clusters in row $\tilde{y}_2$
and in columns $\tilde{x}_1$ and $\tilde{x}_2$
are the ones which the above sequence will solve.
So we need only determine how to correctly interleave
the moves for rows $\tilde{y}_1$ and $\tilde{y}_3$
with the existing move sequence.
We know from the statement of the lemma
that the existing move sequence
satisfies the following constraints:
\begin{align*}
\max_{x_1 \in X_1} I_1(x_1) < \min_{x_1 \in X_1} I_2(x_1) \le \max_{x_1 \in X_1} I_2(x_1) < \min_{x_2 \in X_2} I_2(x_2).
\end{align*}
So to ensure that each cluster $(x_1, \tilde{y}_1)$ is solved,
we insert the two copies of the move $\tilde{y}_1$ to satisfy:
\begin{align*}
\max_{x_1 \in X_1} I_1(x_1) < I_1(\tilde{y}_1) < \min_{x_1 \in X_1} I_2(x_1) \le \max_{x_1 \in X_1} I_2(x_1) < I_2(\tilde{y}_1) < \min_{x_2 \in X_2} I_2(x_2).
\end{align*}
Similarly, we know from the statement of the lemma
that the existing move sequence
satisfies these constraints:
\begin{align*}
\max_{x_2 \in X_2} I_1(x_2) < \min_{x_2 \in X_2} I_2(x_2) \le \max_{x_2 \in X_2} I_2(x_2).
\end{align*}
So we insert the two copies of the move $\tilde{y}_3$
as follows, to ensure that each cluster $(x_2, \tilde{y}_3)$ is solved:
\begin{align*}
\max_{x_2 \in X_2} I_1(x_2) < I_1( \tilde{y}_3) <  \min_{x_2 \in X_2} I_2(x_2) \le \max_{x_2 \in X_2} I_2(x_2) < I_2( \tilde{y}_3).
\end{align*}
To ensure that $I_1(\tilde y_1) < I_2(\tilde y_1) < I_1(\tilde y_3) < I_2(\tilde y_3)$,
we note that the above two constraints do not actually determine
the ordering of $I_2(\tilde y_1)$ and $I_1(\tilde y_3)$.
So we can pick an ordering where $I_2(\tilde y_1) < I_1(\tilde y_3)$,
which will ensure that all clusters are solved.
\qed
\end{proof}
}

\begin{figure}[t]
\centering
\setlength{\unitlength}{0.01in}
\begin{picture}(320, 307)
\put(10, 0){\includegraphics[width=3.0in]{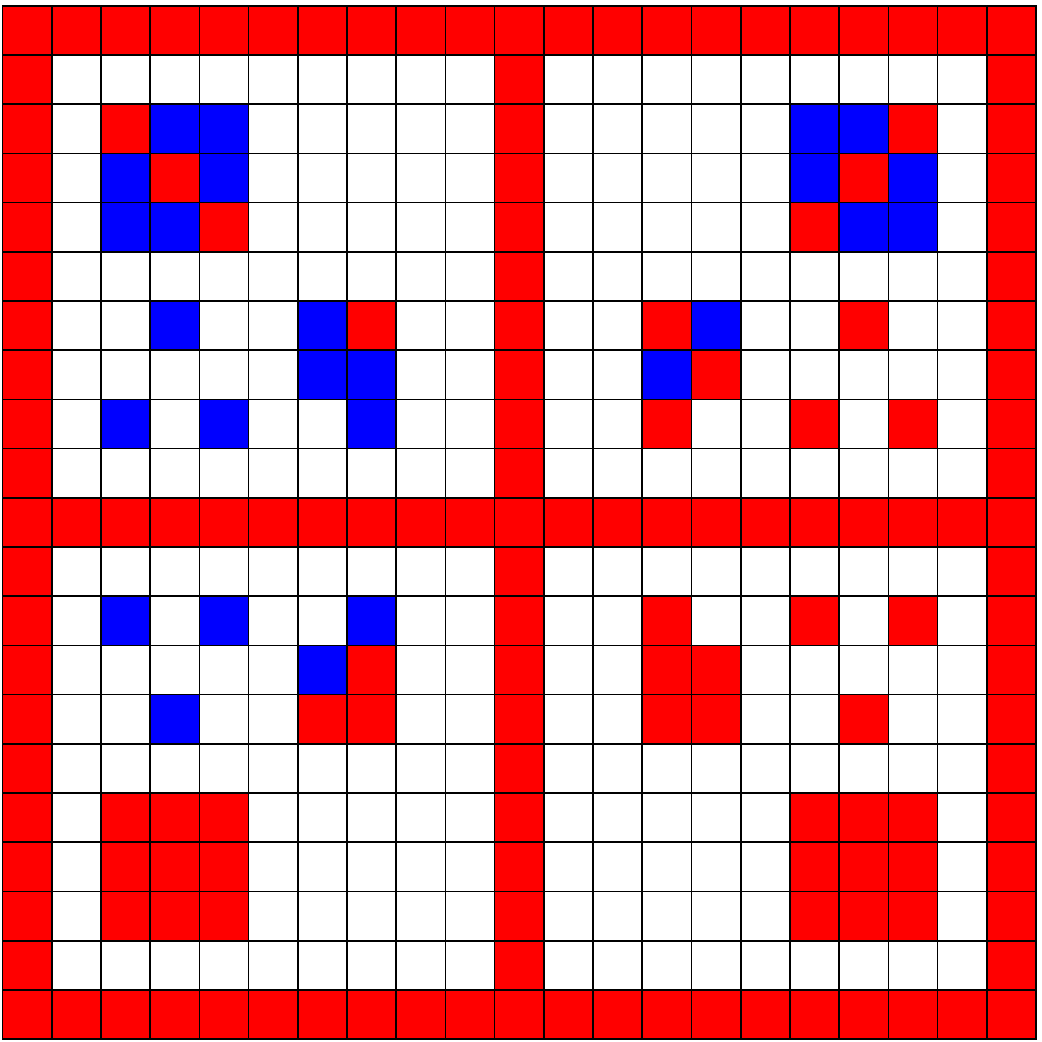}}
\put(41, 303){$x_1$}
\put(55, 303){$x_2$}
\put(69, 303){$x_3$}
\put(97, 303){$\tilde{x}_1$}
\put(111, 303){$\tilde{x}_2$}
\put(-3, 261){$y_1$}
\put(-3, 247){$y_2$}
\put(-3, 233){$y_3$}
\put(-3, 204){$\tilde{y}_1$}
\put(-3, 190){$\tilde{y}_2$}
\put(-3, 176){$\tilde{y}_3$}
\end{picture}
\caption{
A sample of the betweenness gadget
from Lemma~\ref{lemma:gadget-between}.
Important cubies are red (solved) and blue (unsolved).
Unimportant cubies are white.
Any ideal solution 
must either have $I_1(x_1) < I_1(x_2) < I_1(x_3)$
or $I_1(x_3) < I_1(x_2) < I_1(x_1)$.
}
\label{figure:gadget-between}
\end{figure}
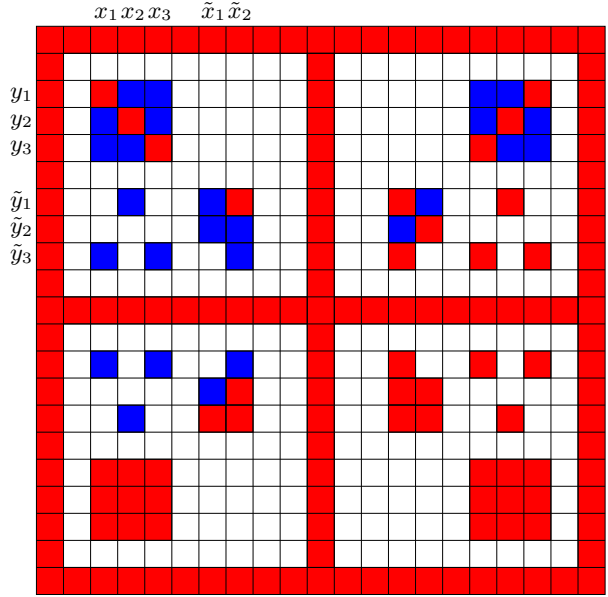

\both{
\begin{lemma}
\label{lemma:gadget-between}
Given three columns $x_1, x_2, x_3 \le \lfloor n / 2 \rfloor$,
there is a gadget
using six extra rows and two extra columns
ensuring that $I_1(x_2)$
lies between $I_1(x_1)$ and $I_1(x_3)$.
As a side effect, this gadget also forces
$I_2(x_2) < I_2(x_1)$, $I_2(x_2) < I_2(x_3)$,
and
\begin{align*}
\max_{x \in \{x_1, x_2, x_3\}} I_1(x) < \min_{x \in \{x_1, x_2, x_3\}} I_2(x).
\end{align*}
\end{lemma}
}

\later{
\begin{proof}
Use a copy of the gadget
from Lemma~\ref{lemma:gadget-before}
to force $I_2(x_2) < I_2(x_1)$ and $I_2(x_2) < I_2(x_3)$.
Let $y_1, y_2, y_3$ be three rows
not used elsewhere in the construction.
Make each \CCS{} in the set
\begin{align*}
\{(x_1, y_2), (x_1, y_3), (x_2, y_1),
(x_2, y_3), (x_3, y_1), (x_3, y_2)\}
\end{align*}
have the configuration depicted in Fig.~\ref{figure:blue-up}.
Make \CCSs{} $(x_1, y_1)$,
$(x_2, y_2)$, and $(x_3, y_3)$
remain solved.
Mark all of these \CCSs{} important.
Now we consider what form
an ideal solution could take,
given these restrictions.

First, consider the case where
$I_1(x_2) < I_1(x_1)$ and $I_1(x_2) < I_1(x_3)$.
Because of the configurations
of \CCSs{} $(x_2, y_1)$ and $(x_2, y_3)$,
we know that
\begin{align*}
I_1(y_1) < I_1(x_2) < I_1(x_1), &&
I_1(y_3) < I_1(x_2) < I_1(x_3), \\
I_2(y_1) < I_2(x_2) < I_2(x_1), &&
I_2(y_3) < I_2(x_2) < I_2(x_3).
\end{align*}
Given these constraints,
and the requirement that
the \CCS{} $(x_1, y_1)$  remain solved,
it must be that $I_2(y_1) < I_1(x_1)$.
Similarly, because the \CCS{} $(x_3, y_3)$
must remain solved,
it must be that $I_2(y_3) < I_1(x_3)$.
In order to solve the \CCSs{}
$(x_1, y_3)$ and $(x_3, y_1)$,
it must be that $I_1(x_1) < I_2(y_3)$
and $I_1(x_3) < I_2(y_1)$.
Therefore
\begin{align*}
I_1(x_1) < I_2(y_3) < I_1(x_3) < I_2(y_1) < I_1(x_1),
\end{align*}
a contradiction.
Hence this case cannot happen.

Now consider the case where
$I_1(x_1) < I_1(x_2)$ and $I_1(x_3) < I_1(x_2)$.
Then we know that the following inequalities hold:
\begin{align*}
I_1(x_1) < I_1(x_2) < I_2(x_2) < I_2(x_1)
&& \textrm{and} &&
I_1(x_3) < I_1(x_2) < I_2(x_2) < I_2(x_3).
\end{align*}
Because of the configurations
of \CCSs{} $(x_2, y_1)$ and $(x_2, y_3)$,
this sandwiching implies that
\begin{align*}
I_1(x_1) < I_2(y_1) < I_2(x_1), &&
I_1(x_1) < I_2(y_3) < I_2(x_1), \\
I_1(x_3) < I_2(y_1) < I_2(x_3), &&
I_1(x_3) < I_2(y_3) < I_2(x_3).
\end{align*}
To ensure that \CCS{} $(x_1, y_1)$ still remains solved,
we need $I_1(x_1) < I_1(y_1)$.
Given the configuration of \CCS{} $(x_1, y_3)$,
we need $I_1(y_3) < I_1(x_1)$.
To ensure that \CCS{} $(x_3, y_3)$ still remains solved,
we need $I_1(x_3) < I_1(y_3)$.
Given the configuration of \CCS{} $(x_3, y_1)$,
we need $I_1(y_1) < I_1(x_3)$.
Thus
\begin{align*}
I_1(y_3) < I_1(x_1) < I_1(y_1) < I_1(x_3) < I_1(y_3),
\end{align*}
a contradiction.
Hence this case is also impossible.

Because neither of the two cases are possible,
$I_1(x_2)$ must lie between $I_1(x_1)$ and $I_1(x_3)$,
which is precisely what we wanted
this gadget to enforce.
We must also show that this gadget
does not enforce any constraints
other than the ones expressed in the lemma.
Given any solution
which satisfies the requirements given in the lemma,
we must be able to insert the moves 
for the new rows and columns
in such a way that all important clusters
will be solved.
The extra constraints on the existing move sequence
ensure the following:
\begin{align*}
\max_{x \in \{x_1, x_3\}} I_1(x) < \min_{x \in \{x_1, x_3\}} I_2(x).
\end{align*}
So we know that we can place the moves for
the extra rows and columns
used by the gadget from Lemma~\ref{lemma:gadget-before}.
We need only determine
where to insert the moves
for the three extra rows added by this gadget.

The constraints given in the statement of the lemma
allow for four different possible orderings
of all of the $x_1, x_2, x_3$ moves.
We consider each case separately.
\begin{enumerate}
\item $I_1(x_1) < I_1(x_2) < I_1(x_3) < I_2(x_2) < I_2(x_1) < I_2(x_3)$.
Then we insert the moves $y_1, y_2, y_3$
so that the following is a subsequence of the move sequence:
\begin{align*}
y_2, y_3, x_1, y_1, x_2, y_3, x_3, y_1, x_2, y_2, x_1, x_3.
\end{align*}
\item $I_1(x_1) < I_1(x_2) < I_1(x_3) < I_2(x_2) < I_2(x_3) < I_2(x_1)$.
Then we insert the moves $y_1, y_2, y_3$
so that the following is a subsequence of the move sequence:
\begin{align*}
y_2, y_3, x_1, y_1, x_2, y_3, x_3, y_1, x_2, y_2, x_3, x_1
\end{align*}
\item $I_1(x_3) < I_1(x_2) < I_1(x_1) < I_2(x_2) < I_2(x_1) < I_2(x_3)$.
Then we insert the moves $y_1, y_2, y_3$
so that the following is a subsequence of the move sequence:
\begin{align*}
y_1, y_2, x_3, y_3, x_2, y_1, x_1, y_3, x_2, y_2, x_1, x_3
\end{align*}
\item $I_1(x_3) < I_1(x_2) < I_1(x_1) < I_2(x_2) < I_2(x_3) < I_2(x_1)$.
Then we insert the moves $y_1, y_2, y_3$
so that the following is a subsequence of the move sequence:
\begin{align*}
y_1, y_2, x_3, y_3, x_2, y_1, x_1, y_3, x_2, y_2, x_3, x_1
\end{align*}
\end{enumerate}
\qed
\end{proof}
}

The \concept{betweenness problem} is a known
NP-hard problem \cite{GareyJohnson79,Opatrny79}.
In this problem,
we are given a set of triples $(a, b, c)$,
and wish to find an ordering on all items
such that, for each triple,
either $a < b < c$ or $c < b < a$.
In other words, for each triple,
$b$ should lie between $a$ and $c$
in the overall ordering.
Lemma~\ref{lemma:gadget-between}
gives us a gadget
which would at first seem
to be perfectly suited
to a reduction from the betweenness problem.
However, because the lemma places
additional restrictions on the order of all moves,
we cannot reduce directly from betweenness.

Instead, we provide a reduction
from another known NP-hard problem,
Not-All-Equal 3-SAT~\cite{GareyJohnson79,Schaefer78}.
In this problem, sometimes known as $\ne$-SAT,
the input is a 3-CNF formula $\phi$
and the goal is to determine whether
there exists an assignment to the variables of $\phi$
such that there is at least one true literal
and one false literal in every clause.
Our reduction from $\ne$-SAT to ideal Rubik solutions
closely follows the reduction from hypergraph 2-coloring
to betweenness \cite{Opatrny79}.

\both{
\begin{theorem}
Given a $\ne$-SAT instance $\phi$,
there exists an $n \times n \times 1$ configuration
and a subset of the cubies
that has an ideal solution
if and only if
$\phi$ has a solution, i.e., belongs to $\ne$-SAT.
\end{theorem}
}

\later{
\begin{proof}
We start with a single column $r_{center} \le \lfloor n / 2 \rfloor$.
For each variable $x_i$ in $\phi$,
we construct two columns
$s_{x_i}, s_{\overline{x_i}} \le \lfloor n / 2 \rfloor$.
We then add a copy of the gadget
from Lemma~\ref{lemma:gadget-between}
to ensure that $I_1(r_{center})$ lies
between $I_1(s_{x_i})$ and $I_1(s_{\overline{x_i}})$,
for each value of $i$.
For each clause $c_j = y_1 \vee y_2 \vee y_3$,
we add a new column $t_j$.
Then we add a copy of the gadget
from Lemma~\ref{lemma:gadget-between}
to ensure that $I_1(t_j)$ lies
between $I_1(s_{y_1})$ and $I_1(s_{y_2})$.
Then we add one more copy of the gadget
from Lemma~\ref{lemma:gadget-between}
to ensure that $I_1(r_{center})$ lies
between $I_1(t_j)$ and $I_1(s_{y_3})$.

Note that the additional constraints
forced by the gadget
from Lemma~\ref{lemma:gadget-between}
mean that all of the following inequalities must hold:
\begin{align*}
I_2(r_{center}) &< I_2(s_{x_i}),
&
I_2(r_{center}) &< I_2(s_{\overline{x_i}}), \\
I_2(t_j) &< I_2(s_{y_1}),
&
I_2(t_j) &< I_2(s_{y_2}), \\
I_2(r_{center}) &< I_2(t_j),
&
I_2(r_{center}) &< I_2(s_{y_3}).
\end{align*}
We can satisfy all of these constraints
by first performing the second move of $r_{center}$,
then performing the second moves of all of the $t$ columns,
and finally by performing the second moves of all of the $s$ columns.
The other constraint imposed by the gadget
from Lemma~\ref{lemma:gadget-between}
can be satisfied by
dividing up the moves into two sequential stages
such that all of the variable, clause, and center moves
are performed exactly once per stage.
Hence, the additional constraints enforced by our gadget
do not affect our ability to construct an ideal solution,
as long as $\phi$ is a member of $\ne$-SAT.

To see why this reduction works,
suppose that $\phi$ is a member of $\ne$-SAT.
We must convert an assignment to the variables of $\phi$
to an ideal solution to the subset of the Rubik's Cube
constructed above.
As noted in the previous paragraph,
we can choose an ordering of all of the second moves
that satisfies the gadgets we have constructed.
To arrange the first moves of all of the columns,
we pick an ordering of
the columns $s$ corresponding to literals
so that $I_1(s_y) < I_1(r_{center})$ for all true literals $y$
and $I_1(r_{center}) < I_1(s_z)$ for all false literals $z$.
The ordering of the literals themselves does not matter,
and can be picked arbitrarily.
This arrangement ensures
that for each $x_i$,
we have either
$I_1(s_{x_i}) < I_1(r_{center}) < I_1(s_{\overline{x_i}})$,
or $ I_1(s_{\overline{x_i}}) < I_1(r_{center}) < I_1(s_{x_i})$;
either way, there will be a way
to correctly arrange the extra columns and rows
used by the gadget from Lemma~\ref{lemma:gadget-between}.

We then must pick times for the first move
for each column $t_j$.
Let $c_j = y_1 \vee y_2 \vee y_3$ be the clause
corresponding to the column we are considering.
\begin{enumerate}
\item Consider the case where $y_1$ and $y_2$
are both true.
Then we pick an arbitrary location for $t_j$
between $s_{y_1}$ and $s_{y_2}$,
so that there is a way to
correctly arrange the extra columns and rows
used by the copy of the gadget
that ensures that $I_1(t_j)$ lies between
$I_1(s_{y_1})$ and $I_1(s_{y_2})$.
This means that $I_1(t_j) < I_1(r_{center})$.
Because $y_1$ and $y_2$ are both true,
then $y_3$ must be false,
and so $I_1(r_{center}) < I_1(y_3)$.
Hence, there will be a way
to correctly arrange the extra columns and rows
used by the copy of the gadget
from Lemma~\ref{lemma:gadget-between}.

\item Now consider the case
where $y_1$ and $y_2$ are both false.
Then we pick an arbitrary location for $t_j$
between $s_{y_1}$ and $s_{y_2}$.
Just as before,
this satisfies the gadget
from Lemma~\ref{lemma:gadget-between}
which forces $I_1(t_j)$ to lie between
$I_1(s_{y_1})$ and $I_1(s_{y_2})$.
This also means that $I_1(r_{center}) < I_1(t_j)$.
Because both $y_1$ and $y_2$ are false,
$y_3$ must be true,
and therefore $I_1(y_3) < I_1(r_{center}) < I_1(t_j)$.
So once again, the gadget is satisfied.

\item Now consider the case
where $y_1$ is true and $y_2$ is false,
or vice versa.
Then $y_3$ is either true or false.
If $y_3$ is true,
then $I_1(y_3) < I_1(r_{center})$,
and so we pick the location for $I_1(t_j)$
to be just slightly larger than $I_1(r_{center})$,
so that $I_1(y_3) < I_1(r_{center}) < I_1(t_j)$
and $I_1(t_j)$ lies between
$I_1(s_{y_1})$ and $I_1(s_{y_2})$.
Similarly, if $y_3$ is false,
then $I_1(r_{center}) < I_1(y_3)$,
and so we pick the location for $I_1(t_j)$
to be just slightly smaller than $I_1(r_{center})$,
so that $I_1(t_j) < I_1(r_{center}) < I_1(y_3)$
and $I_1(t_j)$ lies between
$I_1(s_{y_1})$ and $I_1(s_{y_2})$.
In either case,
the location for the first move $t_j$
will satisfy both of the gadgets created
using Lemma~\ref{lemma:gadget-between}.
\end{enumerate}
So if we have a solution to $\phi$,
then we also have an ideal solution.

Now we wish to prove the converse.
Suppose that we have an ideal solution.
Then we construct a solution to $\phi$
by setting $x_i$ to be true
if and only if $I_1(s_{x_i}) < I_1(r_{center})$.
Because of the gadgets that we constructed,
we know that all true literals $y$
have the property that $I_1(s_y) < I_1(r_{center})$,
while all false literals $z$
have the property that $I_1(r_{center}) < I_1(s_z)$.
To see why this works,
consider a clause $c_j = y_1 \vee y_2 \vee y_3$.
Assume, for the sake of contradiction,
that all three literals in the clause are true.
Then $I_1(s_{y_1}), I_1(s_{y_2}), I_1(s_{y_3}) < I_1(r_{center})$.
Because our betweenness gadget is working correctly,
we know that $I_1(t_j) < I_1(r_{center})$, as well.
This means that $I_1(r_{center})$ does not lie between
$I_1(t_j)$ and $I_1(s_{y_3})$,
which means that our solution could not have been ideal.
So our assumption must be wrong,
and not all of the literals in the clause are true.
A similar argument shows
that not all of the literals in the clause are false.
This means that the clause has
at least one true literal and at least one false literal,
and so $\phi$ is a member of $\ne$-SAT.
\qed
\end{proof}
}


\section{Optimally Solving an $O(1) \times O(1) \times n$ Rubik's Cube}
\label{DP}

For the $c_1 \times c_2 \times n$ Rubik's Cube with $c_1 \neq n \neq c_2$,
the asymmetry of the puzzle
leads to a few additional definitions.
We will call a slice \concept{short}
if the matching coordinate is $z$;
otherwise, a slice is \concept{long}.
A \concept{short move} involves rotating a short slice;
a \concept{long move} involves rotating a long slice.
We define \CCS{} $i$
to be the pair of slices $z = i$ and $z = (n - 1) - i$.
If $n$ is odd,
then \CCS{} $(n - 1) / 2$
will consist of a single slice.
This definition
means that any short move affects
the position and orientation of cubies
in exactly one \CCS{}.

\both{
\begin{lemma}
\label{lemma-short-moves}
Given any $c_1 \times c_2 \times n$ Rubik's Cube configuration,
and any \CCS{} $t$,
the number of short moves affecting that \CCS{}
in the optimal solution
is at most $\shortmoves$.
\end{lemma}
}

\later{
\begin{proof}
Consider the subsequence of moves in the optimal solution
which affect $t$.
This should contain all of the long moves,
and only those short moves which rotate
one of the two slices in $t$.
For notation purposes,
we merge all consecutive long moves together
into compound long moves $L_0, \ldots, L_k$,
so that the sequence of moves is
$L_0 \COMP s_1 \COMP L_1 \COMP s_2 \COMP \ldots \COMP L_{k - 1} \COMP s_{k} \COMP L_k$.
For convenience,
we define $s_0$ to be the identity function,
so that we can write the sequence of moves as
$s_0 \COMP L_0 \COMP s_1 \COMP L_1 \COMP s_2 \COMP \ldots \COMP s_{k} \COMP L_k$.

We define the following to be the results of performing certain moves:
\begin{align*}
\SEQ{i, j} &= s_i \COMP L_i \COMP s_{i + 1} \COMP L_{i + 1} \COMP \ldots \COMP s_{j - 1} \COMP L_{j - 1} \COMP s_j \COMP L_j, \\
\LS{i, j} &= L_i \COMP L_{i + 1} \COMP \ldots \COMP L_{j - 1} \COMP L_j.
\end{align*}
Assume for the sake of contradiction
that there exist $i < j$ such that
$\SEQ{0, i} \COMP \LS{i + 1, k} = \SEQ{0, j} \COMP \LS{j + 1, k}$.
Composing by the inverse of $\LS{j + 1, k}$, we obtain
$\SEQ{0, i} \COMP \LS{i + 1, j} = \SEQ{0, j}$.
Therefore
\begin{align*}
\SEQ{0, i} \COMP \LS{i + 1, j} \COMP \SEQ{j + 1, k}  &= \SEQ{0, j} \COMP \SEQ{j + 1, k} \\
&= \SEQ{0, k}.
\end{align*}
Hence we can omit the moves $s_{i + 1}, \ldots, s_j$
while still ending up with the correct configuration
for $t$.
This means that there exists a sequence of moves,
shorter than the original,
which brings the Rubik's Cube into the same configuration.
But the original sequence was optimal.
Therefore, our assumption must be wrong,
and $\forall i < j$, $\SEQ{0, i} \COMP \LS{i + 1, k} \ne \SEQ{0, j} \COMP \LS{j + 1, k}$.

In both $\SEQ{0, i} \COMP \LS{i + 1, k}$
and $\SEQ{0, j} \COMP \LS{j + 1, k}$,
the set of long moves is the same,
so the configuration of all \CCSs{} other than $t$
must also be the same.
Therefore, the results of those moves
must differ in the configuration
of the \CCS{} $t$.
The \CCS{}
with the greatest number of configurations
is the \CCS{} which contains the ends of the Rubik's Cube
(i.e., $z = 0$ and $z = n - 1$),
because of the additional information given
by the exposed sides of the cube.
For that \CCS{},
the total number of different configurations is
$\le 2^{2c_1c_2} \cdot 4^{4(c_1 + c_2)} = 2^{2c_1c_2 + 8(c_1 + c_2)}$.
Each short move
affecting $t$
in the optimal solution
must lead to a new configuration of $t$,
and so the number of short moves must be $\le \shortmoves$.
\qed
\end{proof}
}

\both{
\begin{lemma}
\label{lemma-tour}
There exists a sequence of long moves
$\ell_1 \COMP \ell_2 \COMP \ldots \COMP \ell_m$,
where $m \le \tourmoves$,
such that:
\begin{enumerate}

\item
$\ell_1 \COMP \ell_2 \COMP \ldots \COMP \ell_m$
is the identity function; and

\item
for every long sequence $L$,
there exists some $i$ such that
$\ell_i \COMP \ell_{i + 1} \COMP \ldots \COMP \ell_m = L$.

\end{enumerate}
\end{lemma}
}

\later{
\begin{proof}
Every sequence of long moves
causes a rearrangement of the cubies in the Rubik's Cube.
However,
there are no long moves
which break up the boxes of cubies
of size $1 \times 1 \times n$
--- each box can be moved or rotated,
but the relative positions
of cubies within the box
are always the same.
There are $c_1c_2$ such boxes,
and each box can be oriented in two ways.
This means that there can be no more than
$(c_1c_2)! \cdot 2^{c_1c_2}$
reachable long configurations.
We treat these long configurations as a graph,
with edges between configurations
that are reachable using a single long move.
If we take a spanning tree of the graph
and duplicate all edges of the tree,
then we can find an Eulerian cycle
which visits all of the nodes in the graph.
If we start at the identity configuration
and move along the cycle,
then we have a path of
$\le \tourmoves$ long moves
satisfying both of the above properties.
\qed
\end{proof}
}

\both{
\begin{lemma}
\label{lemma-long-moves}
Given any $c_1 \times c_2 \times n$ Rubik's Cube configuration,
the number of long moves in the optimal solution is at most $\longmoves$.
\end{lemma}
}

\later{
\begin{proof}
Assume,
for the sake of contradiction,
that the number of long moves in the optimal solution
is greater than $\longmoves$.
Then we can construct another solution
with the same number of short moves as the optimal,
and fewer long moves.
Let $\ell_1 \COMP \ell_2 \COMP \ldots \COMP \ell_m$ be
the sequence of long moves satisfying Lemma \ref{lemma-tour}.
Let $L$ be the sequence of long moves in the optimal solution,
and let $i$ be the index such that
$\ell_i \COMP \ell_{i + 1} \COMP \ldots \COMP \ell_m = L$.
We choose the sequence of long moves in our constructed solution to be
\begin{align*}
(\ell_i \COMP \ell_{i + 1} \COMP \ldots \COMP \ell_m) \COMP
\underbrace{
(\ell_1 \COMP \ell_2 \COMP \ldots \COMP \ell_m) \COMP
\ldots \COMP
(\ell_1 \COMP \ell_2 \COMP \ldots \COMP \ell_m)
}_{\textrm{$\shortmoves$ times}}.
\end{align*}
This long move sequence has the same result as $L$.
We must also specify how to interleave the short moves
with the long move sequence.
For a fixed long move sequence,
the arrangement of short moves for one \CCS{}
does not affect any other \CCS{}.
Consequently, if we can correctly interleave
the short moves for one \CCS{},
we can correctly interleave the short moves
for all \CCSs{}.

Pick an arbitrary \CCS{}.
Consider the subsequence of moves in the optimal solution
which affect said \CCS{}.
For notation purposes,
we merge all consecutive long moves together
into compound long moves $L_0, \ldots, L_k$,
so that the sequence of moves is
$L_0 \COMP s_1 \COMP L_1 \COMP s_2 \COMP \ldots \COMP L_{k - 1} \COMP s_k \COMP L_k$.
We place the short moves
into the above sequence of long moves
starting with $s_k$.
Let $a$ be the index such that
$\ell_a \COMP \ldots \COMP \ell_m = L_k$.
We insert $s_k$ into the $k$th repetition of
$(\ell_1 \COMP \ldots \COMP \ell_m)$
between $\ell_{a - 1}$ and $\ell_{a}$.
This ensures that the sequence of long moves
occurring after $k$ will be equivalent to $L_k$.

In general,
say that we have placed $s_{i + 1}, \ldots, s_k$
in the $(i + 1)$st through $k$th repetitions of
$(\ell_1 \COMP \ldots \COMP \ell_m)$.
Say we want to place $s_i$ in the $i$th repetition.
Let $b$ be the index such that
$s_{i + 1}$ was placed
between $\ell_{b - 1}$ and $\ell_b$.
Let $a$ be the index such that:
\begin{align*}
\ell_a \COMP \ldots \COMP \ell_m
= L_i \COMP (\ell_1 \COMP \ldots \COMP \ell_{b - 1})^{-1}
= L_i \COMP (\ell_{b - 1}^{-1} \COMP \ell_{b - 2}^{-1} \COMP \ldots \COMP \ell_{2}^{-1} \COMP \ell_1^{-1})
\end{align*}
This ensures that the sequence of long moves between
$s_i$ and $s_{i + 1}$ will be:
\begin{align*}
\ell_a \COMP \ldots \COMP \ell_m \COMP \ell_1 \COMP \ldots \COMP \ell_{b - 1}
= L_i \COMP (\ell_1 \COMP \ldots \COMP \ell_{b - 1})^{-1} \COMP (\ell_1 \COMP \ldots \COMP \ell_{b - 1}) = L_i
\end{align*}
Hence,
if the number of long moves
in the optimal solution
is greater than $\longmoves$,
we can create a solution
with the same number of short moves
and fewer long moves.
This means that we have
a contradiction.
\qed
\end{proof}
}

\both{
\begin{theorem}
Given any $c_1 \times c_2 \times n$ Rubik's Cube configuration,
it is possible to find the optimal solution in time polynomial in $n$.
\end{theorem}
}

\later{
\begin{proof}
By Lemma \ref{lemma-long-moves},
we know that the total number of long moves
in the optimal solution is at most $\longmoves$,
which is constant.
We also know that there are a total of
$c_1 + c_2$ possible long moves.
Hence, the total number of possible sequences of long moves is constant,
so we can enumerate all of these in time $O(1)$.

For each of the sequences of long moves,
we want to find the optimal solution using that sequence of long moves.
Because the long moves are fixed,
we can calculate the short moves
for each \CCS{}
independently.
To calculate the short moves for some fixed \CCS{},
we note that between two sequential long moves,
there are at most four different ways to rotate
each of the two slices in the \CCS{},
for a total of at most sixteen possible combinations of short moves.

For a given \CCS{},
we have to consider
$\le 16^{1 + \longmoves}$
possible combinations.
This is constant,
so we can try all possibilities
to see if they solve the \CCS{}.
We can pick the shortest of those.
If we perform this operation
for all \CCSs{},
we will have an optimal solution
for this particular sequence of long moves.
If we calculate this for all sequences of long moves,
then we can pick the overall optimal solution
by taking the sequence of minimum length.
\qed
\end{proof}
}


\iffull
\section{Conclusion and Open Problems}

In this paper, we presented several new results.
First, we introduced a technique
for parallelizing the solution
to two types of generalized Rubik's Cubes.
As a result,
we showed that the diameter of the configuration space
for these two types of Rubik's Cubes
is $\Theta(n^2 / \log n)$.
In addition, we showed that it is NP-hard
to find the shortest sequence of moves
which solves a given subset of the cubies
in an $n \times n \times 1$ Rubik's Cube.
Finally, we showed that there exists a polynomial-time algorithm
for solving Rubik's Cubes with dimensions
$c_1 \times c_2 \times n$, where $c_1 \ne n \ne c_2$.

Our results leave several questions open.
The most obvious questions
concern the NP-hardness result:
whether it can be modified
to show the NP-hardness of
optimally solving the whole $n \times n \times 1$ Rubik's Cube,
and whether it can be further modified
to show the NP-hardness of
optimally solving the whole $n \times n \times n$ Rubik's Cube.
The other questions concern approximation algorithms.
In particular, is there a constant-factor polynomial-time approximation
algorithm for finding an approximately optimal solution sequence
from a given configuration?
The analogous question for the $n^2-1$ puzzle has a positive answer
\cite{Ratner-Warmuth-1990}.
The parallelism techniques we introduced for the diameter results
seem to be central to developing such an approximation algorithm.

\fi


\bibliographystyle{plain}
\bibliography{rubiks}


\appendix
\magicappendix

\end{document}